\documentclass[journal,onecolumn]{IEEEtran}
\usepackage{amsmath,amsthm,mathtools} 
\usepackage{amssymb,amsfonts,mathrsfs,bm} 
\usepackage{extarrows} 
\usepackage{array,tabularx,booktabs,multirow,colortbl,makecell} 
\usepackage{tablefootnote}
\usepackage{tikz-cd} 
\usepackage{physics} 
\usepackage{graphicx} 
\usepackage{xcolor} 
\usepackage[caption=false,font=normalsize,labelfont=sf,textfont=sf]{subfig} 
\usepackage{hyperref} 
\hypersetup{hidelinks,colorlinks=true,linkcolor=blue,citecolor=blue,urlcolor=blue}
\usepackage{stfloats} 
\usepackage{textcomp} 
\usepackage{url}
\usepackage[noadjust]{cite}
\makeatletter
\newtheoremstyle{mythm}{3pt}{3pt}{}{16pt}{\bfseries}{:}{.5em}{}
\theoremstyle{mythm}
\newtheorem{theorem}{Theorem}
\newtheorem{example}{Example}
\newtheorem{definition}{Definition}
\newtheorem{remark}{Remark}

\newtheorem{corollary}{Corollary}
\newtheorem{lemma}{Lemma}

\newtheorem{construction}{Construction}

\newcommand{\cC}{\mathcal{C}}

\newcommand{\cI}{\mathcal{I}}

\newcommand{\cK}{\mathcal{K}}

\newcommand{\cZ}{\mathcal{Z}}

\newcommand{\N}{\mathbb{N}}

\newcommand{\F}{\mathbb{F}}

\newcommand{\sC}{\mathscr{C}}

\DeclareMathOperator{\coff}{Coff}
\DeclareMathOperator{\diag}{Diag}
\renewcommand{\le}{\leqslant}

\renewcommand{\ge}{\geqslant}

\newcommand{\mb}[1]{\mathbf{#1}}        

\newcommand{\parenv}[1]{\left( #1 \right)}

\newcommand{\anglenv}[1]{\left< #1 \right>}

\begin{document}
\title{Convertible Polynomial Evaluation Codes in the Merge Regime: A Skew-Polynomial Framework}
\author{
    Songping~Ge,
    Han~Cai,~\IEEEmembership{Member,~IEEE},
    and~Xiaohu~Tang,~\IEEEmembership{Fellow,~IEEE}
    \thanks{Songping Ge and Xiaohu Tang are with the School of Mathematics, Southwest Jiaotong University, Chengdu, 610031, China. Xiaohu Tang is also with  the Information Coding and Transmission Key Lab of Sichuan Province, CSNMT Int. Coop. Res. Centre (MoST), Southwest Jiaotong University, Chengdu, 610031, China  (e-mail: gesongping@my.swjtu.edu.cn; xhutang@swjtu.edu.cn).}
    \thanks{Han Cai is with the Information Coding and Transmission Key Lab of Sichuan Province, CSNMT Int. Coop. Res. Centre (MoST), Southwest Jiaotong University, Chengdu, 610031, China (e-mail: hancai@aliyun.com).}
    \thanks{Parts of this work have been accepted for presentation at 2026 IEEE International Symposium on Information Theory \cite{ge2026convertible}.}
}
\maketitle

\begin{abstract}
In this paper, we develop an algebraic framework for merge-regime convertible codes that works directly with the polynomial-evaluation structure of the underlying codes. In the skew polynomial ring $\mathbb{F}_{q^m}[x;x^q,0]$, congruence relations do not in general translate directly into the pointwise evaluation identities required for code conversion, since skew-polynomial evaluation is not multiplicative. To overcome this difficulty, we characterize the minimal skew polynomials associated with unions of conjugacy classes, establish an evaluation-compatible product rule, and derive a special Chinese Remainder Theorem (sCRT) tailored to these evaluation sets. Based on this machinery, we propose two conversion templates for skew polynomial evaluation codes (PECs). The first applies to distinct initial PECs, whose different evaluation sets naturally provide the moduli required by sCRT. The second treats identical initial PECs, for which distinct auxiliary evaluation sets and explicit algebraic compatibility conditions are introduced to preserve unchanged symbols and to generate written symbols from designated read symbols. When instantiated with the standard basis $1,x,\ldots,x^{k-1}$, both constructions yield merge conversions for linearized Reed--Solomon codes whose final codes are equivalent to linearized Reed--Solomon codes and achieve per-symbol access-optimal cost. The framework further specializes to the commutative ring $\mathbb{F}_q[x]$, yielding corresponding conversion constructions for ordinary PECs and recovering the known polynomial-form constructions for Reed--Solomon and Tamo--Barg codes as special cases. As a further application, to the best of our knowledge, this specialization gives the first merge-regime convertible construction with Gabidulin initial codes and a final code equivalent to a Gabidulin code, while attaining per-symbol optimal access under the symbol-access model considered in this paper.
\end{abstract}

\section{introduction}\label{sec-intro}

The widespread adoption of large-scale cloud storage and distributed file systems, such as Azure Storage \cite{huang2012Erasure} and the Google File System \cite{ghemawat2003Google}, has made disk failures a routine part of operations rather than a rare exception. Erasure codes with fixed code rates are used to ensure data reliability. 
However, the failure rate of devices may change during their lifetime. Thus, usually we design codes for the worst-case scenarios, which may result in a large amount of redundancy for periods with low failure rates. To overcome this problem and achieve more economical storage systems, in \cite{kadekodi2019Cluster},
Kadekodi et al. proposed a new  framework capable of adjusting code rates and redundancy on the fly based on changes in device failure rates, which helps maximize storage efficiency throughout the device lifetime. This process is known as \textit{code conversion} \cite{maturana2022Convertible}, i.e., the process converts an \textit{initial code} into a \textit{final code}. Later, Maturana and Rashmi \cite{maturana2022Convertible} introduced \textit{convertible codes} to enable efficient conversions between initial and final codes. In \cite{kong2024Locally}, Kong extended this scheme to the codes with locality. In \cite{ge2026mds}, Ge et al. generalized the definition of convertible codes by allowing initial codes and final codes to be distinct codes.

However, code conversion may incur a large access cost, which is defined as the total number of symbols read from the
initial code plus the number of new symbols generated from those reads during the conversion process. Thus, for convertible codes,
one of the main objectives is to reduce the access cost consumed in conversion. On one hand, \cite{maturana2022Convertible,ge2026mds, shi2026Bounds,wang2026Lower} proposed theoretical bounds on the access cost for convertible codes. On the other hand, explicit constructions are proposed to achieve the bounds with equality: Code conversion on maximum distance separable (MDS) codes \cite{maturana2022Convertible, maturana2020Accessoptimal, chopra2024Low, ge2026mds}; Code conversion on locally repairable codes \cite{maturana2023Locally, kong2024Locally, ge2026Locally, shi2026Bounds}; Code conversion on Reed-Muller codes \cite{Gruica2026Linear}; Code conversion on vector codes \cite{maturana2022Bandwidth, maturana2023Bandwidth, wang2026Lower, singhvi2025Tight}.

From the known bounds on the access cost for code conversion on MDS codes, each new symbol in the final code must be generated by reading at least one symbol from each initial code. In the literature, the special case where each new symbol is generated by reading exactly one symbol from each initial code is referred to as \textit{per-symbol access-optimal MDS convertible
code}~\cite{ramkumar2026MDS}.

Existing explicit constructions of convertible codes can be broadly viewed from two main perspectives. One line of work develops conversions based on generator matrices \cite{maturana2022Convertible,maturana2020Accessoptimal,chopra2024Low,Gruica2026Linear}, parity-check matrices \cite{ge2026mds,ge2026Locally,shi2026Bounds} or rational function fields \cite{shi2026Bounds}. Another line exploits polynomial representations of structured codes, with existing constructions mainly focusing on Reed--Solomon and Tamo--Barg codes \cite{kong2024Locally}. For Reed--Solomon and Tamo--Barg codes, the existing constructions impose different algebraic conditions on their evaluation vectors, such as span-preserving and structured transformation conditions. Since both code families admit a common polynomial-evaluation representation, a natural question is: can these polynomial-based constructions be unified and extended to general polynomial evaluation codes (PECs)?

A natural way to pursue this idea is through the Chinese Remainder Theorem (CRT). In the commutative setting, one may associate different evaluation sets with suitable annihilating polynomials and combine the message polynomials of the initial codewords through congruence relations. However, extending this approach to skew polynomial evaluation is not straightforward. In general, skew-polynomial evaluation is not multiplicative, i.e., $(fg)(a)$ need not equal $f(a)g(a)$. Consequently, a polynomial congruence modulo an annihilating polynomial does not by itself guarantee the pointwise evaluation identities needed to preserve unchanged symbols or to generate written symbols during conversion. The main algebraic challenge is therefore to identify moduli associated with the evaluation sets for which both CRT-type congruences and the required evaluation behavior can be controlled.

To address these issues, our main contributions are summarized as follows.

\begin{itemize}

\item \textbf{Algebraic foundation for CRT-based conversion over skew polynomials.}
We study the skew polynomial ring $\mathbb{F}_{q^m}[x;x^q,0]$ and characterize the minimal skew polynomials associated with individual conjugacy classes and unions of conjugacy classes. In particular, these polynomials have a commutative structure that allows us to establish an evaluation-compatible product rule for the corresponding factors. Based on these properties, we derive a special Chinese Remainder Theorem (sCRT) whose moduli are determined by the evaluation sets. This provides the algebraic mechanism that connects polynomial congruences with the pointwise evaluations required in code conversion.

\item \textbf{Two conversion constructions for skew polynomial evaluation codes.}
Using the above sCRT machinery, we develop two merge-regime conversion constructions for skew polynomial evaluation codes (PECs). The first construction treats distinct initial skew PECs, where the different evaluation sets naturally provide distinct minimal skew polynomials that can serve as the moduli in the sCRT construction. The second construction treats identical initial skew PECs, where such distinct moduli are not directly available. To handle this case, we introduce auxiliary evaluation sets together with two explicit algebraic compatibility conditions, C1 and C2. These conditions ensure that the unchanged symbols of the identical initial codes can be consistently represented in the final code and that each written symbol can be generated from the prescribed read symbols.

\item \textbf{Per-symbol access-optimal conversions for linearized Reed--Solomon codes.}
When the polynomial basis is chosen as
$1,x,\ldots,x^{k-1}$, the two skew-PEC constructions specialize to linearized Reed--Solomon codes. We obtain conversions for both distinct and identical initial linearized Reed--Solomon codes in which the final code is equivalent to a linearized Reed--Solomon code. Moreover, the resulting conversions attain the per-symbol access lower bound: each written symbol is generated by reading exactly one symbol from each initial codeword.

\item \textbf{Commutative specialization and a Gabidulin-code application.}
When the skew polynomial framework specializes to the commutative ring $\mathbb{F}_q[x]$, the sCRT reduces to the standard CRT and the two skew-PEC constructions yield corresponding conversion constructions for ordinary PECs. This viewpoint recovers the known polynomial-form constructions for Reed--Solomon and Tamo--Barg codes as special cases of the same framework. More importantly, by choosing the polynomial basis associated with Gabidulin codes, we obtain, to the best of our knowledge, the first merge-regime convertible construction whose initial codes are Gabidulin codes and whose final code is equivalent to a Gabidulin code. The construction also achieves per-symbol optimal access under the symbol-access model considered in this paper.
\end{itemize}

Our objective is therefore structural rather than a field-size improvement for generic MDS convertible codes: the main gain is a common polynomial-evaluation mechanism that preserves or recovers specific algebraic code families. The parameter and field-size restrictions in Table~\ref{compare mds}, including the divisibility requirements of the skew constructions, should be viewed as the cost of this structured framework rather than as improvements over the smallest known fields for unrestricted MDS conversion.

We list our main results together with representative known MDS merge-conversion constructions in Table~\ref{compare mds}, where each initial code $\mathcal{C}_{I_i}$ is an $[n_{I_i},k_{I_i}]_q$ MDS code for $i\in[t]$, and the final code $\mathcal{C}_{F}$ is an $[n_{F},k_{F}]$ MDS code. Denote $r_{I_i}\triangleq n_{I_i}-k_{I_i}$ and $r_{F}\triangleq n_{F}-k_{F}$.

The paper is organized as follows. Section~\ref{sec-pre} reviews irregular convertible codes and the access-cost model. Section~\ref{sec-CRT} develops the evaluation-compatible sCRT machinery for $\mathbb{F}_{q^m}[x;x^q,0]$. Section~\ref{sec-conversion LRS} gives the two skew-PEC conversion templates, Constructions~\ref{con C} and~\ref{con D}, and specializes them to linearized Reed--Solomon codes. Section~\ref{sec-degraded} specializes the framework to $\mathbb{F}_q[x]$, yielding Constructions~\ref{con A} and~\ref{con B}, and then derives the Reed--Solomon, Tamo--Barg, and Gabidulin consequences. Section~\ref{sec con} concludes the paper and discusses limitations and open directions.

\renewcommand{\arraystretch}{1.5}
\begin{table}[t!]
    \label{compare mds}
    \centering
    \caption{Some explicit known $(t,1)_{|\mathbb{F}|}$ access-optimal irregular merge-convertible codes for MDS codes to an MDS code}
    \resizebox{0.8\textwidth}{!}{
    \begin{tabular}{|c|c|c|c|c|}
    \hline
     \textbf{Construction} & \textbf{Type} & \textbf{Restrictions}  & \textbf{Field size $|\mathbb{F}|$ requirement} & \textbf{Reference} \\
    \hline
    \hline
     \multirow{2}{*}{General} & \multirow{2}{*}{generator matrix} & \multirow{2}{*}{\makecell{$\mathcal{C}_{I_1}=\mathcal{C}_{I_2}=\cdots=\mathcal{C}_{I_t}$,\\ $r_F\le\min\{k_{I_1},r_{I_1}\}$}} &  \multirow{2}{*}{$|\mathbb{F}|\ge \max\left\{ \makecell{2^{O((n_F)^3)},\\ n_{I_1}-1} \right\}$} & \multirow{2}{*}{\cite[Theorem 21]{maturana2022Convertible}}  \\
    & & & &\\
    \hline
    \multirow{2}{*}{Hankel-I} & \multirow{2}{*}{generator matrix} & \multirow{2}{*}{\makecell{$\mathcal{C}_{I_1}=\mathcal{C}_{I_2}=\cdots=\mathcal{C}_{I_t}$,\\ $r_F\le \lfloor r_{I_1}/t\rfloor$}} &  \multirow{2}{*}{$|\mathbb{F}|\ge \max\{n_F,n_{I_1}\}-1$} &  \multirow{2}{*}{\cite[Example 22]{maturana2022Convertible}} \\
    & & & &\\
    \hline
     \multirow{2}{*}{Hankel-II} & \multirow{2}{*}{generator matrix} & \multirow{2}{*}{\makecell{$\mathcal{C}_{I_1}=\mathcal{C}_{I_2}=\cdots=\mathcal{C}_{I_t}$,\\ $r_F\le r_{I_1}-t+1$}} &  \multirow{2}{*}{$|\mathbb{F}|\ge \max\{ k_{I_1}r_{I_1},n_{I_1}-1\}$} & \multirow{2}{*}{\cite[Example 23]{maturana2022Convertible}}  \\
    & & & &\\
    \hline
     \multirow{4}{*}{\makecell{Hankel$_s$ for\\ $t\le s\le r_{I_1}$}} & \multirow{4}{*}{generator matrix} & \multirow{4}{*}{\makecell{$\mathcal{C}_{I_1}=\mathcal{C}_{I_2}=\cdots=\mathcal{C}_{I_t}$,\\ $r_F\le(s-t+1)\lfloor r_{I_1}/s \rfloor$\\ $+\max\{(r_{I_1}\mod s)-t+1,0\}$}} & \multirow{4}{*}{$|\mathbb{F}|\ge \max\left\{ \makecell{sk_{I_1}+\lfloor r_{I_1}/s \rfloor-1,\\ n_{I_1}-1} \right\}$} & \multirow{4}{*}{\cite[Theorem 24]{maturana2022Convertible}} \\
    & & & &\\
    & & & &\\
    & & & &\\
    \hline
    \multirow{3}{*}{\makecell{General \\ $k \times r$ Super-Regular \\ Vandermonde Matrix}} & \multirow{3}{*}{generator matrix} & \multirow{3}{*}{\makecell{$\mathcal{C}_{I_1}=\mathcal{C}_{I_2}=\cdots=\mathcal{C}_{I_t}$,\\ $k_F \le k$,\\ $r_{I_1}=r_F\le r$}} & \multirow{3}{*}{\makecell{$|\mathbb{F}|\ge \max\left\{ \makecell{2^{O(r_F \log k_F)},\\ n_{I_1}-1} \right\}$}} & \multirow{3}{*}{\cite[Theorem 3]{chopra2024Low}}  \\
    & & & &\\
    & & & &\\
    \hline
    \multirow{3}{*}{\makecell{$k \times 3$ Super-Regular \\ Vandermonde Matrix}} & \multirow{3}{*}{generator matrix} & \multirow{3}{*}{\makecell{$\mathcal{C}_{I_1}=\mathcal{C}_{I_2}=\cdots=\mathcal{C}_{I_t}$,\\ $k_F \le k$,\\ $r_{I_1}=r_F\le 3$}} & \multirow{3}{*}{\makecell{$|\mathbb{F}|=p^w$,\\ $w\ge k$}} & \multirow{3}{*}{\cite[Theorem 4]{chopra2024Low}}  \\
     && & &\\
     && & &\\
    \hline
    \multirow{3}{*}{\makecell{$k \times 3$ Super-Regular \\ Vandermonde Matrix\\ for Finite Fields of Characteristic $2$}} & \multirow{3}{*}{generator matrix} & \multirow{3}{*}{\makecell{$\mathcal{C}_{I_1}=\mathcal{C}_{I_2}=\cdots=\mathcal{C}_{I_t}$,\\ $k_F \le k$,\\ $r_{I_1}=r_F\le 3$}} & \multirow{3}{*}{$|\mathbb{F}|=2^w \ge k$} & \multirow{3}{*}{\cite[Theorem 5]{chopra2024Low}}  \\
     && & &\\
     && & &\\
    \hline
     \multirow{2}{*}{GRS} & \multirow{2}{*}{polynomial form} & \multirow{2}{*}{\makecell{$\mathcal{C}_{I_1}=\mathcal{C}_{I_2}=\cdots=\mathcal{C}_{I_t}$,\\ $r_{F}\le\min\{k_{I_1},r_{I_1}\}$}} & \multirow{2}{*}{\makecell{$|\mathbb{F}|\ge (t+1)\max\{k_{I_1},r_{I_1}\}+1$,\\ $\max\{k_{I_1},r_{I_1}\}\mid(q-1)$}} & \multirow{2}{*}{\cite[Corollary II.2]{kong2024Locally}}  \\
     && & &\\
    \hline
    Extended GRS & parity-check matrix &  / & $|\mathbb{F}|\ge \max\{n_{I_1},\cdots,n_{I_{t}},n_{F}\}-1$ &  \cite[Theorem 6]{ge2026mds}  \\
    \hline
    Cauchy Matrix &  parity-check matrix & $\mathcal{C}_{I_1}, \dots, \mathcal{C}_{I_t}$ are $t$ distinct codes & $|\mathbb{F}|\ge k_{F} +2\max\{r_{I_1},\cdots,r_{I_{t}}, r_{F}\}$ & \cite[Theorem 7]{ge2026mds}  \\
    \hline
    \multirow{2}{*}{Rational Function Field} & \multirow{2}{*}{Function-field form} & \multirow{2}{*}{$t\le r_F$} & \multirow{2}{*}{\makecell{$|\mathbb{F}|\ge \max\{n_{I_1},\cdots,n_{I_{t}},n_{F}\}-1$,\\ exists a subgroup of order $r_F$ in $\mathrm{PGL_2}(|\mathbb{F}|)$\tablefootnote{$\mathrm{PGL_2(q)}$ is the projective general linear group, i.e., $\mathrm{PGL_2(q)}=\mathrm{GL}_2(q)/Z(\mathrm{GL}_2(q))$, where $Z(\mathrm{GL}_2(q))=\{u\mathbf{I}_2:u\in\F_q^*\}$ and $\mathrm{GL}_2(q)$ is the general linear group of $2\times2$ invertible matrices over $\F_q$. For more details, please refer to \cite{shi2026Bounds}.} }} & \multirow{2}{*}{\cite[Theorem IV.1]{shi2026Bounds}} \\
    && & &\\
    \hline
    \multirow{3}{*}{Linearized Reed-Solomon Code} & \multirow{3}{*}{polynomial form} & \multirow{3}{*}{\makecell{$\mathcal{C}_{I_1}, \dots, \mathcal{C}_{I_t}$ are $t$ distinct $[n_{I_1},k_{I_1}]$ codes,\\ $m\mid n_{I_1}, n_{F}, k_{I_1}$,\\ $r_{F}\le\min\{k_{I_1},r_{I_1}\}$}} & \multirow{3}{*}{\makecell{$|\mathbb{F}|\ge \left( \frac{tk_{I_1}+r_{I_1}}{m}+1 \right)^m$}} & \multirow{3}{*}{Corollaries~\ref{cor-con C}}  \\
    & & & &\\
    & & & &\\
    \hline
    \multirow{3}{*}{Linearized Reed-Solomon Code} & \multirow{3}{*}{polynomial form} & \multirow{3}{*}{\makecell{$\mathcal{C}_{I_1}=\mathcal{C}_{I_2}=\cdots=\mathcal{C}_{I_t}$, \\ $m\mid n_{I_1}, n_{F}, k_{I_1}$,\\ $r_{F}\le\min\{k_{I_1},r_{I_1}\}$}} & \multirow{3}{*}{\makecell{$|\mathbb{F}|\ge \left( \frac{tk_{I_1}+r_{I_1}}{m}+1 \right)^m$}} & \multirow{3}{*}{Corollaries~\ref{cor5}}  \\
    & & & &\\
    & & & &\\
    \hline
    \multirow{2}{*}{Gabidulin Code} & \multirow{2}{*}{polynomial form} & \multirow{2}{*}{\makecell{$\mathcal{C}_{I_1}=\mathcal{C}_{I_2}=\cdots=\mathcal{C}_{I_t}$,\\ $r_{F}\le\min\{k_{I_1},r_{I_1}\}$}} & \multirow{2}{*}{$|\mathbb{F}|\ge q_1^{(k_{I_1}+1)(m+1)}, k_{I_1}+t-2\ge m> t$} & \multirow{2}{*}{Corollary \ref{thm6}} \\
    & & & &\\
    \hline
    \end{tabular}
   }
\end{table}

\section{preliminaries}\label{sec-pre}

We first introduce some notation and definitions.
For any $a\in\N$, we denote $[a]\triangleq \{1,2,\dots,a\}$.
Let $\F$ be the finite field. For any prime power $q$, we denote $\F_q$ be the finite field with $q$ elements and $\F_q^*\triangleq \F_q\setminus\{0\}$. The norm for $\alpha$ in a finite field $\F_{q^m}$ defined as $\mathrm{N}_{\F_{q^m}/\F_q}(\alpha)=\alpha^{1+q+\cdots+q^{m-1}}$.
An $[n,k]_q$ linear code $\cC$ over $\F_q$ is a $k$-dimensional subspace of $\F_q^n$ with $k\times n$ generator matrix $\mb{G}$ such that $\cC= \{\mb{mG}:\mathbf{m}\in\F_q^k\}$. An $[n,k]_q$ linear code $\cC$ is called an MDS code if and only if the minimum Hamming distance of $\cC$ is $n-k+1$. We say a code $\mathcal{C}_1$ with a generator matrix $\mathbf{G}_1$ is equivalent to a code $\mathcal{C}_2$ with a generator matrix $\mathbf{G}_2$, if $\mathbf{G}_1=\mathbf{G}_2\mathbf{D}$ for some invertible diagonal matrix $\mathbf{D}$. For a matrix $\mb{M} \in \F_q^{m\times n}$, denote $\anglenv{\mb{M}}$ as the vector space spanned by the column vectors of $\mb{M}$ over $\F_q$. For a vector $\mb{v}=(v_1,v_2,\dots,v_n)$, we denote $(v_i)_{i\in[n]} \triangleq \mb{v}$. Let $\top$ be the transpose operator.

\subsection{Convertible codes}

We recall irregular convertible codes~\cite{ge2026mds} in the merge regime, which generalized the framework for code conversions proposed by \cite{maturana2022Convertible} that concentrates on the identical initial codes.

\begin{definition}[Irregular convertible code in the merge regime]\label{def cc}
Let $t>1$ be a positive integer. A $(t,1)_q$ irregular convertible code $\sC$ over $\F_q$ consists of: (1) $t$ initial codes $\cC_{I_1},\dots,\cC_{I_{t}}$ and a final codes $\cC_{F}$, where $\cC_{I_i}$ is an $[n_{I_i},k_{I_i}]_q$ code for $i\in[t]$, and $\cC_{F}$ is an $[n_{F},k_{F}]_q$ code, satisfying $\sum_{i=1}^{t}k_{I_i} = k_{F}$; (2) an irregular conversion procedure $\varphi$ which is an injective map as $\varphi: \prod_{i=1}^{t} \cC_{I_i} \to \cC_{F}$, i.e., $\varphi( (c_{\mb{m}_1},\dots,c_{\mb{m}_{t}}) )
= c_{\mb{m}}$, where $c_{\mb{m}_i} \in \cC_{I_i}$ is encoded from message $\mb{m}_i \in \F_q^{k_{I_i}}$ for $i\in[t]$, and $c_{\mb{m}} \in \cC_{F}$ is encoded from message $\mb{m} = (\mb{m}_1,\dots,\mb{m}_{t}) \in \F_q^{k_{F}}$.
\end{definition}

The code symbols in the initial codes are classified into three categories: (1) \textit{unchanged symbols} are code symbols in the initial codes that remain in the final codes; (2) \textit{read symbols} are code symbols in the initial codes that participate in the conversion process; (3) \textit{retired symbols} are code symbols in the initial codes that are discarded after the conversion process. The code symbols in the final codes are classified into two categories: (1) \textit{unchanged symbols} are same as above; (2) \textit{written symbols} are code symbols in the final codes that are not unchanged symbols.

\begin{definition}
The \textit{read access cost} of an irregular convertible code is defined as the total number of read symbols, denoted as $\rho_r$. The \textit{write access cost} of an irregular convertible code is defined as the total number of written symbols, denoted as $\rho_w$. The \textit{access cost} of an irregular convertible code is the sum of its read and write access costs, denoted as $\rho \triangleq \rho_r+\rho_w$.
\end{definition}

The lower bound on access cost for \textit{MDS irregular convertible codes} in the merge regime, i.e., the initial and final codes are MDS codes, is as follows.

\begin{theorem}[\cite{ge2026mds}]\label{thm bound}
Let $\sC$ be a $(t,1)_q$ MDS irregular merge-convertible code, then the write access cost $\rho_w$ satisfies $\rho_w\ge r_F$ and the read access cost $\rho_r$ satisfies
\begin{equation*}
\rho_r \ge \sum_{ i\in[t], r_{F}\le \min\{k_{I_i},r_{I_i}\} }  r_{F} + \sum_{ i\in[t], r_{F}> \min\{k_{I_i},r_{I_i}\} }k_{I_i},
\end{equation*}
where $r_{F}=n_{F}-k_{F}$ and $r_{I_i}=n_{I_i}-k_{I_i}$ for $i\in[t]$.
\end{theorem}

We recall the definition of per-symbol access-optimal MDS convertible codes~\cite{ramkumar2026MDS}.

\begin{definition}[Per-symbol access-optimal MDS convertible code~\cite{ramkumar2026MDS}]\label{def per}
Let $r_F\le\min\{k_{I_i},r_{I_i}\}$ for $i\in[t]$. An access-optimal $(t,1)_q$ MDS irregular convertible code is said to be per-symbol access-optimal if each of the $r_F$ written symbols in the final codeword can be computed independently by reading exactly one symbol from each of the $t$ initial codewords.
\end{definition}

\section{Evaluation-Compatible sCRT for $\mathbb{F}_{q^m}[x;x^q,0]$}\label{sec-CRT}

This section develops the algebraic mechanism that makes a CRT-based conversion possible in the skew setting. In a commutative polynomial ring, CRT congruences can be evaluated using the ordinary product rule. For skew polynomials, the corresponding implication is not automatic because evaluation is generally not multiplicative. Our goal is therefore not to establish a CRT for arbitrary skew rings, but to identify a family of moduli tied to the evaluation sets for which congruence and pointwise evaluation interact in the form needed by code conversion.

We proceed in three steps before stating the sCRT in Subsection~\ref{subsec-CRT}. Subsection~\ref{subsec-skew} reviews skew polynomial rings and evaluation. Subsection~\ref{subsec-conjugacy} identifies, through conjugacy classes, a commutative annihilating polynomial and proves the product rule used later to translate polynomial products into pointwise evaluations. Subsection~\ref{subsec-minimal} shows that these annihilating polynomials are precisely the minimal skew polynomials of the relevant conjugacy classes and extends the characterization to unions of classes. These results provide evaluation-set-specific moduli for the sCRT and explain why the theorem can be used directly in the conversion constructions of Section~\ref{sec-conversion LRS}.

\subsection{Skew polynomial ring}\label{subsec-skew}

In this subsection, we introduce the definition of the skew polynomial ring and propose some basic properties for commutativity and evaluation, which will be widely used in the subsequent subsections.

\begin{definition}
An endomorphism $\sigma$ of $\mathbb{F}$ is a map $\sigma:\F\to\F$ such that $\sigma(a+b)=\sigma(a)+\sigma(b)$ and $\sigma(ab)=\sigma(a)\sigma(b)$ for all $a,b\in\F$. A $\sigma$-derivation $\delta$ is a map $\delta:\F\to\F$ such that $\delta(a+b)=\delta(a)+\delta(b)$ and $\delta(ab)=\sigma(a)\delta(b)+\delta(a)b$ for all $a,b\in\F$.
\end{definition}

\begin{definition}[Skew polynomial ring~\cite{ore1933theory}]
Let $\sigma$ be an endomorphism of $\F$ and $\delta$ be a $\sigma$-derivation. The skew polynomial ring $\F[x;\sigma,\delta]$ is a non-commutative ring of skew polynomials as
\begin{equation*}
    \F[x;\sigma,\delta] = \left\{ \sum_{i\in[n]} a_ix^i : n\in\mathbb{Z}, a_i\in\F \right\},
\end{equation*}
where the coefficients are written to the left of variable $x$ and  addition is component wise and multiplication is distributive subject to the rule, for all $\alpha\in\F$,
\begin{equation}\label{rule}
    x \cdot \alpha = \sigma(\alpha)x+\delta(\alpha).
\end{equation}
The degree of $f(x)=\sum_{i\in[n]}a_ix^i\in\F[x;\sigma,\delta]$ is defined by the maximum $i$ such that $a_i\ne0$, denoted as $\deg(f)$. In particular, define $\deg(0)=\infty$ \cite{ore1933theory}. 
\end{definition}

The definition implies that $\deg(fg)=\deg(f)+\deg(g)$ for any two skew polynomials $f,g\in\F[x;\sigma,\delta]$. 
The multiplication rule in \eqref{rule} shows that $\alpha\cdot x$ is not always equal to $x\cdot\alpha$, implying that $\F[x;\sigma,\delta]$ is a non-commutative ring.
However, in some special cases, skew polynomials may commute with others.

\begin{lemma}\label{lem1}
Let $\F[x;\sigma,\delta]=\F_{q^m}[x;x^q,0]$. Then, $x^m \cdot \alpha = \alpha \cdot x^m$ for all $\alpha\in\F_{q^m}$, and $x \cdot \alpha = \alpha \cdot x$ for all $\alpha\in\F_q$.
\end{lemma}

\begin{IEEEproof}
By the multiplication rule in \eqref{rule}, $x^m\cdot \alpha= \sigma^m(\alpha)\cdot x^m=\alpha^{q^m}\cdot x^m= \alpha\cdot x^m$ for $\alpha\in\F_{q^m}$. Also, we have $x\cdot \alpha=\sigma(\alpha)\cdot x=\alpha^q\cdot x=\alpha\cdot x$ for $\alpha\in\F_{q}$.
\end{IEEEproof}

Recall that $\F[x;\sigma,\delta]$ are left and right Euclidean domains~\cite{ore1933theory}. The evaluation of skew polynomials is defined as follows.

\begin{definition}[Evaluation~\cite{lam1985general,gopi2022improved}]
The evaluation of a skew polynomial $f(x)\in\F[x;\sigma,\delta]$ at a point $\alpha\in\F$, denoted by $f(\alpha)$, is defined as the remainder obtained when divide $f$ by $x-\alpha$ on the right, i.e., $f(x)=q(x)(x-\alpha)+f(\alpha)$.
\end{definition}

The evaluations of monomials of form $x^i$ shown in \cite{gopi2022improved} is useful for our purposes.

\begin{definition}[Power functions~\cite{gopi2022improved}]\label{def power}
The power functions are defined inductively. For every $\alpha\in\F$, $\Phi_0(\alpha)=1$ and $\Phi_{i+1}(\alpha)=\sigma(\Phi_i(\alpha))\alpha +\delta(\Phi_i(\alpha))$.
\end{definition}

\begin{lemma}[\cite{gopi2022improved}]\label{lem2}
Let $f(x)=\sum_{i}a_ix^i\in \F[x;\sigma,\delta]$. Then $f(\alpha)=\sum_{i}a_i\Phi_i(\alpha)$.
\end{lemma}

The norm  $\mathrm{N}_{\F_{q^m}/\F_q}(\alpha)$ and the power functions have the following relationship.

\begin{lemma}\label{lem3}
If $f(x)=x^m\in\F_{q^m}[x;x^q,0]$, then $f(\alpha)=\mathrm{N}_{\F_{q^m}/\F_q}(\alpha)$ for all $\alpha\in\F_{q^m}$.
\end{lemma}

\begin{IEEEproof}
$f(\alpha)=\Phi_m(\alpha)$ by Lemma \ref{lem2}. Since $\delta=0$ and $\sigma=x^q$, we have $\Phi_m(\alpha)= \alpha^{1+q+\cdots+q^{m-1}}=\mathrm{N}_{\F_{q^m}/\F_q}(\alpha)$ by Definition \ref{def power}.
\end{IEEEproof}

\subsection{Conjugacy and Product Rule}\label{subsec-conjugacy}

The multiplication rule in \eqref{rule} implies the rule $(fg)(x)=f(x)g(x)$ is not always true over the skew polynomial ring. In this subsection, we show commutative skew polynomials by the conjugacy class in Lemma \ref{lem4}, which can deduce $(fg)(x)=f(x)g(x)$ in Lemma \ref{lem7}. At the beginning, we introduce the definition of the conjugacy class as follows.

\begin{definition}[Conjugacy~\cite{lam1985general,martinez-penas2018skew}]
Let $\alpha\in\F$ and  $D_\alpha^{\sigma,\delta}:\F\to\F$ is a $(\sigma,\delta)$-operator~\cite{leroy1995pseudo} such that
\begin{equation*}
    D_\alpha^{\sigma,\delta}(\beta) = \sigma(\beta)\alpha+\delta(\beta).
\end{equation*}  Then the conjugacy class of $\alpha$ is defined as
\begin{equation*}
    C(\alpha) = \{ D_\alpha^{\sigma,\delta}(\beta)\beta^{-1} : \beta\in\F^* \},
\end{equation*}
Denote $D_\alpha$ for simplicity when $\sigma,\delta$ are known from the context.
\end{definition}

The conjugacy relation is an equivalence relation in $\F$, and conjugacy classes give a partition of $\F$ as shown in \cite{martinez-penas2018skew}. We will now show a special annihilating skew polynomial for the conjugacy class $C(\alpha)$, which has the property of commutativity.

\begin{lemma}[Commutativity]\label{lem4}
Let $\F[x;\sigma,\delta]=\F_{q^m}[x;x^q,0]$. Then the skew polynomial $f(x)=x^m-\mathrm{N}_{\F_{q^m}/\F_q}(\alpha)$ satisfies 
\begin{itemize}
    \item [1)] $f(a)=0$ for all $a\in C(\alpha)$;
    \item [2)] $f(a)=\mathrm{N}_{\F_{q^m}/\F_q}(\alpha_1)-\mathrm{N}_{\F_{q^m}/\F_q}(\alpha)\in\mathbb{F}_{q}^{*}$ for all $a\in C(\alpha_1)$, where $\alpha_1\not\in C(\alpha)$;
    \item [3)] $f(x)g(x)=g(x)f(x)$ for any skew polynomial $g(x)\in\F_{q^m}[x;x^q,0]$.
\end{itemize}
\end{lemma}

\begin{IEEEproof}
According to Lemma \ref{lem3}, $f(a)= \mathrm{N}_{\F_{q^m}/\F_q}(a)- \mathrm{N}_{\F_{q^m}/\F_q}(\alpha)$ for all $a\in\F$. Let $a\in C(\alpha)$. Then, there exists $\beta\in\F_{q^m}^*$ such that $a=D_\alpha(\beta)\beta^{-1}=\alpha\beta^{q-1}$ by $\sigma=x^q$ and $\delta=0$. Thus,
\begin{equation*}
\begin{aligned}
    \mathrm{N}_{\F_{q^m}/\F_q}(a)
    &= \mathrm{N}_{\F_{q^m}/\F_q}(\alpha\beta^{q-1}) \\
    &= \mathrm{N}_{\F_{q^m}/\F_q}(\alpha) \cdot \mathrm{N}_{\F_{q^m}/\F_q}(\beta^{q-1}) \\
    &= \mathrm{N}_{\F_{q^m}/\F_q}(\alpha) \cdot \beta^{q^m-1} \\
    &= \mathrm{N}_{\F_{q^m}/\F_q}(\alpha).
\end{aligned}
\end{equation*}

Similarly, we claim $f(a)= \mathrm{N}_{\F_{q^m}/\F_q}(\alpha_1)- \mathrm{N}_{\F_{q^m}/\F_q}(\alpha)\in\F_q$ for $a\in C(\alpha_1)$.
Now for the second part, it is sufficient to show $f(a)\ne0$. Assume that
\begin{equation*}
    \mathrm{N}_{\F_{q^m}/\F_q}(\alpha_1)
    =\mathrm{N}_{\F_{q^m}/\F_q}(\alpha).
\end{equation*}
Then, $\mathrm{N}_{\F_{q^m}/\F_q}(\alpha_1\alpha^{-1})=1$. This is to say 
that $\alpha_1\alpha^{-1}=\theta^{c(q-1)}$ with a primitive element $\theta$ of $\F_{q^m}$, i.e.,
\begin{equation*}
    \alpha_1 = \theta^{c(q-1)}\alpha \in C(\alpha).
\end{equation*}
Therefore, $f(a)\ne0$ since $\alpha_1\not\in C(\alpha)$.

At last, the commutative property comes from Lemma \ref{lem1} as
\begin{equation*}
\begin{aligned}
    f(x)g(x)
    &= (x^m-\mathrm{N}_{\F_{q^m}/\F_q}(\alpha))g(x) \\
    &= x^m \cdot g(x) - \mathrm{N}_{\F_{q^m}/\F_q}(\alpha) \cdot g(x) \\
    &= g(x) \cdot x^m -  g(x) \cdot \mathrm{N}_{\F_{q^m}/\F_q}(\alpha) \\
    &= g(x)(x^m-\mathrm{N}_{\F_{q^m}/\F_q}(\alpha)) \\
    &= g(x)f(x).
\end{aligned}
\end{equation*}

Above all, we complete the proof.
\end{IEEEproof}

In fact, Lemma \ref{lem4}  implies a necessary and sufficient condition to determine whether $\alpha$ and $\beta$ are in the same conjugacy class.

\begin{corollary}\label{cor2}
Let $\F[x;\sigma,\delta]=\F_{q^m}[x;x^q,0]$. For any $\alpha,\beta\in\F_{q^m}$, $\mathrm{N}_{\F_{q^m}/\F_q}(\beta) = \mathrm{N}_{\F_{q^m}/\F_q}(\alpha)$ if and only if $\alpha,\beta$ are in the same conjugacy class.
\end{corollary}

\begin{IEEEproof}
The sufficiency is directly from the proof of Lemma \ref{lem4}-2). The necessity holds by Lemma \ref{lem4}-1).
\end{IEEEproof}

In what follows, we show that the commutative skew polynomial $f(x)=x^m-\mathrm{N}_{\F_{q^m}/\F_q}(\alpha)$ in Lemma \ref{lem4} satisfies the property of evaluation, i.e., $(gf)(\alpha)=g(\alpha)f(\alpha)$ for any skew polynomial $g(x)\in \F_{q^m}[x;x^q,0]$ and any element $\alpha\in\mathbb{F}_{q^m}$. 
 We first introduce the general property of evaluation of $(fg)(\alpha)$ as follows.

\begin{lemma}[Product rule~\cite{lam1985general,lam1988vandermonde}]\label{lem6}
Let $f,g\in\F[x;\sigma,\delta]$. If $f(\alpha)=0$, then $(gf)(\alpha)=0$. If $f(\alpha)\ne0$, then
\begin{equation*}
    (gf)(\alpha)=g\left( D_\alpha(f(\alpha)) \cdot f(\alpha)^{-1} \right) \cdot f(\alpha).
\end{equation*}
\end{lemma}

The product rule is an important lemma for skew polynomials. For the proof of Lemma \ref{lem6} the reader may refer to \cite[Appendix B]{gopi2022improved}. In what follows, we propose a crucial conclusion about the skew polynomial $f(x)=x^m-\mathrm{N}_{\F_{q^m}/\F_q}(\alpha)$ for $C(\alpha)$, which can simplify the evaluation of the product of polynomials.

\begin{lemma}\label{lem7}
Let $\F[x;\sigma,\delta]=\F_{q^m}[x;x^q,0]$. For any skew polynomial $g(x)\in\F_{q^m}[x;x^q,0]$ and any conjugacy class $C(\alpha)$, the product $(gf)(x)$, where $f(x)=x^m-\mathrm{N}_{\F_{q^m}/\F_q}(\alpha)$ for $C(\alpha)$, satisfies
\begin{equation*}
    (gf)(a) = g(a)f(a) \mbox{ for all } a\in\F.
\end{equation*}
\end{lemma}

\begin{IEEEproof}
If $f(a)=0$, then $(gf)(a)= 0= g(a)f(a)$ by Lemma \ref{lem6}. If $f(a)\ne0$, also by Lemma \ref{lem6}, we have
\begin{equation*}
\begin{aligned}
    (gf)(a)
    &= g(D_a(f(a)) \cdot f(a)^{-1}) \cdot f(a) \\
    &= g( (f(a)^{q} \cdot a) \cdot f(a)^{-1} ) \cdot f(a) \\
    &= g( f(a)^{q-1} \cdot a ) \cdot f(a) \\
    &= g(a)f(a),
\end{aligned}
\end{equation*}
where the last equation comes from $f(a)^{q-1}=1$ since $f(a)\in\F_q$ by Lemma \ref{lem4}-2) and $f(a)\ne0$ by assumption.
\end{IEEEproof}

\subsection{Minimal skew polynomial and centralizer}\label{subsec-minimal}

In previous Subsection \ref{subsec-conjugacy}, we have proposed the annihilating skew polynomial $f(x)=x^m-\mathrm{N}_{\F_{q^m}/\F_q}(\alpha)$ for conjugacy class $C(\alpha)$. In this subsection, we further show that $f(x)$ is the minimal skew polynomial of $C(\alpha)$ in Corollary \ref{cor4}, whose function is the same as that of co-prime polynomials in sdandard CRT for the ring $\mathbb{F}[x]$ shown in next Subsection \ref{subsec-CRT}.

We begin by introducing some necessary concepts, i.e., minimal skew polynomial, P-base, and centralizer. This is because, the minimal skew polynomial of the P-base of $C(\alpha)$ is simultaneously the minimal skew polynomial of $C(\alpha)$ shown in Corollaries \ref{cor3} and \ref{cor4}, and the P-base can be given by the centralizer shown in Lemma \ref{lem11}.

\begin{definition}
Let $A\subseteq\F[x;\sigma,\delta]$ and $B\subseteq\F$. The zero set of $A$ is defined as
\begin{equation*}
    \cZ(A) = \{ a\in\F : f(a)=0, \forall f\in A \}.
\end{equation*}
The associated ideal of $B$ is defined as
\begin{equation*}
    \cI(B) = \{ f\in\F[x;\sigma,\delta] : f(a)=0, \forall a\in B \}.
\end{equation*}
\end{definition}

As shown in \cite{martinez-penas2018skew}, $\cI(B)$ is a left ideal in $\F[x;\sigma,\delta]$ for any $B\subseteq\F$. And there exists a unique monic skew polynomial of minimal degree among those in $\cI(B)$, denoted as $\Lambda_B\in\cI(B)$ which is called the \textit{minimal skew polynomial} of $B$. We will now introduce the P-base.

\begin{definition}[P-base~\cite{lam1988vandermonde,martinez-penas2018skew}]\label{def p base}
Let $B\subseteq\F$.
\begin{enumerate}
  \item The P-closure of $B$ is defined as $\overline{B}=\cZ(\cI(B))=\cZ(\Lambda_B)$. And $B$ is called P-closed if $\overline{B}=B$.
  \item Let $B$ be P-closed. A subset $U\subseteq B$ generates $B$ if $\overline{U}=B$, which is called as set of P-generators of $B$.
  \item An element $a\in\F$ is P-independent from $B$ if $a\not\in\overline{B}$. The set $B$ is called P-independent if every $a\in B$ is P-independent from $B\setminus\{a\}$.
  \item Let $B$ be P-closed. A subset $U\subseteq B$ is called a P-basis of $B$ if $U$ is P-independent and a set of P-generators of $B$.
\end{enumerate}
\end{definition}

To analyze the P-basis of the conjugacy class $C(\alpha)$, we first show that $C(\alpha)$ is P-closed for $\sigma=x^q$ and $\delta=0$.

\begin{lemma}\label{lem8}
Let $\F[x;\sigma,\delta]=\F_{q^m}[x;x^q,0]$. Any conjugacy class $C(\alpha)$ for $\alpha\in\F$ is P-closed, i.e., $\overline{C(\alpha)}=C(\alpha)$.
\end{lemma}

\begin{IEEEproof}
By the definition of P-closed, we need to prove $\cZ(\cI(C(\alpha))) = C(\alpha)$. According to Lemma \ref{lem4},
\begin{equation*}
    f(x)=x^m-\mathrm{N}_{\F_{q^m}/\F_q}(\alpha) \in \cI(C(\alpha))
\end{equation*}
Thus,
\begin{equation*}
    \cZ(\cI(C(\alpha))) \subseteq \cZ(f).
\end{equation*}
By Corollary \ref{cor2},
\begin{equation*}
    \cZ(f) = C(\alpha).
\end{equation*}
Therefore,
\begin{equation}\label{eqn25}
    \cZ(\cI(C(\alpha))) \subseteq C(\alpha).
\end{equation}
In addition, let $a\in C(\alpha)$. Then, $g(a)=0$ for any $g(x)\in\cI(C(\alpha))$. This implies that $a\in\cZ(\cI(C(\alpha)))$. Thus,
\begin{equation}\label{eqn26}
    C(\alpha) \subseteq \cZ(\cI(C(\alpha))).
\end{equation}
Combining \eqref{eqn25} and \eqref{eqn26}, we complete the proof.
\end{IEEEproof}

Furthermore, to determine an explicit P-basis of $C(\alpha)$, we recall some basic result and conceptions proposed in \cite{gopi2022improved} and \cite{martinez-penas2018skew}.

\begin{definition}[Centralizer~\cite{lam1988vandermonde}]
The $(\sigma,\delta)$-centralizer of $\alpha\in\F$, or simply centralizer, is defined as
\begin{equation*}
     \cK_\alpha^{\sigma,\delta} = \{ \beta\in\F : D_\alpha^{\sigma,\delta}(\beta)\beta^{-1} = \alpha \}.
\end{equation*}
Denote $\cK_\alpha$ for simplicity when $\sigma,\delta$ are known from the context.
\end{definition}

\begin{lemma}[{\cite[Example 1]{gopi2022improved}}]\label{lem9}
Let $\F[x;\sigma,\delta]=\F_{q^m}[x;x^q,0]$. The centralizer of every element $\alpha\in\F_{q^m}^*$ is $\cK_a=\F_q$.
\end{lemma}

\begin{lemma}[\cite{lam1985general,martinez-penas2018skew}]\label{lem10}
A subset $B\subseteq\F$ is P-independent if and only if $\deg(\Lambda_B)=|B|$.
\end{lemma}

\begin{lemma}[{\cite[Cor. 27]{martinez-penas2018skew}}]\label{lem11}
Fix $\alpha\in\F$, let $\beta_i\in\F^*$ and define $\gamma_i=D_\alpha(\beta_i)\beta_i^{-1}$ for $i\in[n]$. Then, $\{\beta_1,\beta_2,\dots,\beta_n\}$ is a basis of the right vector space over $\cK_\alpha$ if and only if $\{\gamma_1,\gamma_2,\dots,\gamma_n\}$ is a P-basis of some P-closed set $B\subseteq\F$.
\end{lemma}

We are ready to show the P-basis of the conjugacy class $C(\alpha)$.

\begin{corollary}\label{cor3}
Let $\F[x;\sigma,\delta]=\F_{q^m}[x;x^q,0]$. For any $\alpha\in\F_{q^m}^*$, let $\{\beta_1,\dots,\beta_m\}$ be a basis of $\F_{q^m}$ over $\F_q$, then $U=\{\gamma_i=\beta_i^{q-1}\alpha : i\in[m]\}$ is a P-basis of the conjugacy class $C(\alpha)$. Moreover, $\Lambda_U(x)= x^m-\mathrm{N}_{\F_{q^m}/\F_q}(\alpha)$.
\end{corollary}

\begin{IEEEproof}
According to Lemma \ref{lem11}, $U$ is a P-basis of some P-closed set. It is sufficient to prove $\overline{U}=C(\alpha)$. By Lemma \ref{lem10} and each P-basis is P-independent from Definition \ref{def p base}, we have
\begin{equation*}
    \deg(\Lambda_U)=m.
\end{equation*}
Note that $f(x)=x^m-\mathrm{N}_{\F_{q^m}/\F_q}(\alpha)$ satisfies
\begin{equation*}
    f(\gamma_i)=0
\end{equation*}
for $i\in[m]$ since $U\subseteq C(\alpha)$ and Lemma \ref{lem4}. Thus, by $\deg(f)=\deg(\Lambda_U)$ and $f(x)$ is a monic skew polynomial, we have
\begin{equation*}
    \Lambda_U(x)=f(x).
\end{equation*}
According to the definition of P-closed in Definition \ref{def p base},
\begin{equation*}
\begin{aligned}
    \overline{U} &= \cZ(\Lambda_U) \\
    &= \{ a\in\F : \Lambda_U(a)=0 \} \\
    &= \{ a\in\F : f(a)=0 \} \\
    &= \{ a\in\F : \mathrm{N}_{\F_{q^m}/\F_q}(a) = \mathrm{N}_{\F_{q^m}/\F_q}(\alpha) \}.
\end{aligned}
\end{equation*}
By Corollary \ref{cor2},
\begin{equation*}
    \overline{U} = C(\alpha).
\end{equation*}
This is to say $U$ is a P-basis of $C(\alpha)$.
\end{IEEEproof}

Moreover, the minimal skew polynomial $\Lambda_U$ of the P-basis $U$ in Corollary \ref{cor3} is also the  minimal skew polynomial $\Lambda_{C(\alpha)}$ of the conjugacy class $C(\alpha)$, i.e., $\Lambda_U=\Lambda_{C(\alpha)}$.

\begin{corollary}\label{cor4}
Let $\F[x;\sigma,\delta]=\F_{q^m}[x;x^q,0]$. For any $\alpha\in\F_{q^m}^*$, $\Lambda_{C(\alpha)}(x)= x^m-\mathrm{N}_{\F_{q^m}/\F_q}(\alpha)$.
\end{corollary}

\begin{IEEEproof}
Assume that $\Lambda_{C(\alpha)} = f(x)\in\F[x;\sigma,0]$ with $\deg(f)<m$. Then, $f(\gamma)=0$ for all $\gamma$ in the P-base $U\subseteq C(\alpha)$. Therefore, $\deg(\Lambda_U)\le\deg(f)<m$ which contradicts that $\deg(\Lambda_U)=m$ by Lemma \ref{lem10}.
\end{IEEEproof}

Finally, we generalize Corollaries \ref{cor3} and \ref{cor4} to the union of multiple conjugacy classes in Theorem \ref{thm7}, which is crucial for sCRT shown in next Subsection \ref{subsec-CRT}.

\begin{theorem}\label{thm7}
Let $\F[x;\sigma,\delta]=\F_{q^m}[x;x^q,0]$. For $\ell$ distinct conjugacy classes $C(\alpha_1),C(\alpha_2),\dots,C(\alpha_\ell)$, let $\{\beta_1,\beta_2,\dots,\beta_m\}$ be a basis of $\F_{q^m}$ over $\F_q$, then $U=\cup_{i\in[\ell]}U_i$ is a P-basis of $C=\cup_{i\in[\ell]}C(\alpha_i)$, where $U_i=\{\gamma_{i,j}=\beta_j^{q-1}\alpha_i : j\in[m]\}$ is a P-basis of $C(\alpha_i)$ for $i\in[\ell]$. Moreover,
\begin{equation*}
    \Lambda_{C}(x)
    = \Lambda_{U}(x)
    = \prod_{i\in[\ell]} \left( x^m-\mathrm{N}_{\F_{q^m}/\F_q}(\alpha_i) \right),
\end{equation*}
where the terms can be arbitrarily reordered.
\end{theorem}

\begin{IEEEproof}
It is sufficient to prove the conclusion on the case $\ell=2$. The conclusion on the case $\ell>2$ can be obtained by induction. For ease of notation, denote $\mathrm{N}(x)=\mathrm{N}_{\F_{q^m}/\F_q}(x)$. According to Lemma \ref{lem4},
\begin{equation*}
    (x^m-\mathrm{N}(\alpha_1))(x^m-\mathrm{N}(\alpha_2))
    = (x^m-\mathrm{N}(\alpha_2))(x^m-\mathrm{N}(\alpha_1)).
\end{equation*}
Therefore, in what follows we only need to prove one order of the product.

We first prove that $\Lambda_{U}=(x^m-\mathrm{N}(\alpha_1))(x^m-\mathrm{N}(\alpha_2))$. Assume that
\begin{equation*}
    g(x)\in \cI(U)
\end{equation*}
with $\deg(g)<2m$. Then, $g(x)\in\cI(U_1)$ since $U_1\subseteq U$. By Corollary \ref{cor3}, $\Lambda_{U_1}(x)=x^m-\mathrm{N}(\alpha_1)$. Thus, there exists a non-zero skew polynomial $g_1(x)\in\F[x;\sigma,0]$ such that
\begin{equation}\label{eqn27}
    g(x) = g_1(x) \Lambda_{U_1}(x).
\end{equation}
Similarly, for $U_2$, there exists a non-zero skew polynomial $g_2(x)\in\F[x;\sigma,0]$ such that
\begin{equation}\label{eqn28}
    g(x) = g_2(x) \Lambda_{U_2}(x).
\end{equation}
Let $a=(\mathrm{N}(\alpha_2)-\mathrm{N}(\alpha_1))^{-1}\in\F_q^*$ and $b=-(\mathrm{N}(\alpha_2)-\mathrm{N}(\alpha_1))^{-1}\in\F_q^*$. Then,
\begin{equation}\label{eqn29}
    a \Lambda_{U_1}(x) + b \Lambda_{U_2}(x) = 1.
\end{equation}
Combining \eqref{eqn27}, \eqref{eqn28} and \eqref{eqn29}, we have
\begin{equation*}
    a \Lambda_{U_1}(x) \cdot g_2(x) \Lambda_{U_2}(x)
    + b \Lambda_{U_2}(x) \cdot g_1(x) \Lambda_{U_1}(x) = g(x).
\end{equation*}
According to Lemma \ref{lem4},
\begin{equation}\label{eqn commut}
\begin{aligned}
    \Lambda_{U_1}(x) \cdot g_2(x) &= g_2(x) \cdot \Lambda_{U_1}(x), \\
    \Lambda_{U_2}(x) \cdot g_1(x) &= g_1(x) \cdot \Lambda_{U_2}(x).
\end{aligned}
\end{equation}
Thus,
\begin{equation*}
\begin{aligned}
    g(x) &= a g_2(x) \Lambda_{U_1}(x) \Lambda_{U_2}(x)
    + b g_1(x) \Lambda_{U_2}(x) \Lambda_{U_1}(x) \\
    &= (ag_2(x)+bg_1(x))(\Lambda_{U_1}(x) \Lambda_{U_2}(x)).
\end{aligned}
\end{equation*}
We claim that $ag_2(x)+bg_1(x)\not \equiv0$. Otherwise, $ag_2(x)+bg_1(x)\equiv0$. Then,
\begin{equation*}
    g_2(x)= -a^{-1}bg_1(x) = g_1(x).
\end{equation*}
Together with \eqref{eqn27} and \eqref{eqn28}, by Lemma \ref{lem4}, we have
\begin{equation*}
    (\Lambda_{U_1}(x)- \Lambda_{U_2}(x))g_2(x) \equiv 0.
\end{equation*}
Since $g_2(x)$ is a non-zero skew polynomial, $\Lambda_{U_1}(x)- \Lambda_{U_2}(x)\equiv0$, i.e.,
\begin{equation*}
    x^m-\mathrm{N}(\alpha_1) = x^m-\mathrm{N}(\alpha_2).
\end{equation*}
This is to say
\begin{equation*}
    \mathrm{N}(\alpha_1) = \mathrm{N}(\alpha_2),
\end{equation*}
which implies a contradiction $C(\alpha_1)=C(\alpha_2)$ by Corollary \ref{cor2}. Hence, $ag_2(x)+bg_1(x)\not \equiv0$ and
\begin{equation*}
    \deg(g)\ge \deg(\Lambda_{U_1})+\deg(\Lambda_{U_2})= 2m,
\end{equation*}
which contradicts to assumption $\deg(g)<2m$. Therefore, $\Lambda_{U}=\Lambda_{U_1}(x) \Lambda_{U_2}(x)$. This is to say $U$ is P-independent by Lemma \ref{lem10}.

Assume that $h(x)\in\cI(C)$ with $\deg(h)<2m$. Then, $h(x)\in\cI(U)$ since $U_1\subseteq C(\alpha_1)$ and $U_2\subseteq C(\alpha_2)$, which contradicts to $\deg(\Lambda_{U})=2m$. Thus, $\deg(\Lambda_{C})\ge2m$. In fact $\Lambda_{C}=\Lambda_{U}$ since $\Lambda_{U}\in\cI(C)$ by Lemmas \ref{lem4} and \ref{lem7}. Therefore,
\begin{equation*}
    \Lambda_{C}= \Lambda_{U}= (x^m-\mathrm{N}(\alpha_1))(x^m-\mathrm{N}(\alpha_2)).
\end{equation*}

Finally, we need to prove $U$ is a P-basis of $C$. According to Definition \ref{def p base} and $U$ is P-independent, we need to prove
\begin{enumerate}
    \item $\overline{C}=C$, i.e., $C$ is P-closed and
    \item $\overline{U}=C$, i.e., $U$ is a P-generator of $C$.
\end{enumerate}
It is sufficient to prove $\overline{U}=C$, since $\overline{U}=C$ implies $C$ is P-closed.  Note that
\begin{equation*}
    \overline{U} = \cZ(\Lambda_U) =  \{ a\in\F : \Lambda_{U}(a)=0 \}.
\end{equation*}
By $\Lambda_{C}=\Lambda_{U}$, we have
\begin{equation*}
    C\subseteq \overline{U}.
\end{equation*}
We claim that $\overline{U}\subseteq C$. Assume that $\Lambda_{U}(a)=0$ for $a\in\F_{q^m}$. Then,
\begin{equation*}
    \Lambda_{U}(a)
    = \Lambda_{U_1}(a) \Lambda_{U_2}(a) = 0.
\end{equation*}
This implies $\Lambda_{U_1}(a)=0$ or $\Lambda_{U_2}(a)=0$ since $\Lambda_{U_1}(a), \Lambda_{U_2}(a)\in \F_q$ by Lemma \ref{lem4}-2). According to Corollary \ref{cor2}, $a\in C(\alpha_1)$ or $a\in C(\alpha_2)$. Thus, $a\in C(\alpha_1)\cup C(\alpha_2)$. This is to say
\begin{equation*}
    \overline{U}\subseteq C.
\end{equation*}
Therefore, $\overline{U}=C$. We complete the proof.
\end{IEEEproof}

\subsection{sCRT for skew polynomials}\label{subsec-CRT}

We now combine the preceding minimal-polynomial characterization and evaluation product rule with skew Lagrange interpolation. The resulting theorem should be viewed as an evaluation-compatible CRT specialization for the moduli $\Lambda_{C(\alpha_i)}(x)$: besides solving the congruence system, its construction is aligned with the P-bases on which the code symbols are evaluated. Lemma~\ref{lem Lag} supplies the interpolation step used in the existence proof.

\begin{lemma}[Lagrange interpolation \cite{lam1985general}]\label{lem Lag}
    Let $B\subseteq\mathbb{F}$ be a P-closed set with P-basis $U=\{\gamma_1,\gamma_2,\dots,\gamma_n\}$. For every $a_1,a_2,\dots,a_n\in\mathbb{F}$, there exists a unique $f(x)\in\mathbb{F}[x;\sigma,\delta]$ such that $\deg(f)<n$ and $f(\gamma_i)=\alpha_i$ for each $i\in[n]$.
\end{lemma}

\begin{theorem}[sCRT]\label{thm8}
Let $\F[x;\sigma,\delta]=\F_{q^m}[x;x^q,0]$. Suppose $C(\alpha_i), i\in[t]$ are $t$ distinct conjugacy classes. Then, for any $t$ skew polynomials $m_1(x),\dots,m_t(x)\in \F_{q^m}[x;x^q,0]$, there exists a unique polynomial $f(x)\in \F_{q^m}[x;x^q,0]$ with respect to the right modulo $\Lambda_{\cup_{i\in[t]}C(\alpha_i)}(x)$ such that
\begin{equation*}
    f(x) \equiv m_i(x) \pmod{\Lambda_{C(\alpha_i)}(x)} \mbox{ for all } i\in[t],
\end{equation*}
where ``mod'' means right modulo since the minimal skew polynomial of $C(\alpha)$ is the generator of left ideal $\cI(C(\alpha))$.
\end{theorem}

\begin{IEEEproof}
First, we prove existence. Let $M_i(x)=\prod_{j\in[t]\setminus\{i\}} \Lambda_{C(\alpha_j)}(x)$ for $i\in[t]$. According to Lemma \ref{lem Lag}, there exists a unique polynomial $N_i(x)\in\F_{q^m}[x;x^q,0]$ with degree less than $m$, such that 
\begin{equation*}
    N_i(\gamma_{i,j}) = (M_i(\gamma_{i,j}))^{-1},
\end{equation*}
where $U_i=\{\gamma_{i,1},\gamma_{i,2},\dots,\gamma_{i,m}\}$ is P-basis of $C(\alpha_i)$ for $i\in[t]$. Then, the polynomial $(N_iM_i)(x)-1$ satisfies
\begin{equation*}
    (N_iM_i)(\gamma_{i,j})-1 = N_i(\gamma_{i,j})M_i(\gamma_{i,j})-1 = 0
\end{equation*}
for all $j\in[m]$ by Lemma \ref{lem7}. Thus, 
\begin{equation*}
    (N_iM_i)(x)-1 \in \mathcal{I}(U_i).
\end{equation*}
Since the minimal skew polynomial $\Lambda_{U_i}(x)$ of $\mathcal{I}(U_i)$ and $\Lambda_{C(\alpha_i)}(x)=\Lambda_{U_i}(x)$ by Corollaries \ref{cor3} and \ref{cor4}, we have
\begin{equation*}
    (N_iM_i)(x)-1 = h(x)\Lambda_{C(\alpha_i)}(x)
\end{equation*}
for some polynomial $h(x)\in\F_{q^m}[x;x^q,0]$. This is to say
\begin{equation*}
    (N_iM_i)(x) \equiv 1 \pmod{\Lambda_{C(\alpha_i)}(x)}
\end{equation*}
and
\begin{equation*}
    (N_iM_i)(x) \equiv 0 \pmod{\Lambda_{C(\alpha_j)}(x)} \text{ for }j\ne i.
\end{equation*}
Let
\begin{equation}\label{eqn41}
    f(x)=\sum_{i\in[t]} (m_iN_iM_i)(x)
\end{equation}
Therefore, $f(x)$ is a solution of system of congruence equations.

Second, we show uniqueness. Assume that $f(x)$ and $g(x)$ satisfy the required system of congruence equations. Then,
\begin{equation*}
    f(x)-g(x) \equiv 0 \pmod{\Lambda_{C(\alpha_i})(x)} \mbox{ for all } i\in[t].
\end{equation*}
Thus, there exists $v_i(x)\in\F[x;\sigma,0]$ for each $i\in[t]$ such that
\begin{equation*}
    f(x)-g(x) = v_i(x)\Lambda_{C(\alpha_i)}(x).
\end{equation*}
Same as the proof of Theorem \ref{thm7}, i.e., \eqref{eqn27} and \eqref{eqn28}, we can derive that
\begin{equation*}
    f(x)-g(x) \equiv 0 \pmod{\prod_{i=1}^{t}\Lambda_{C(\alpha_i})(x)}.
\end{equation*}
Note that $\prod_{i=1}^{t}\Lambda_{C(\alpha_i})(x)=\Lambda_{\cup_{i\in[t]}C(\alpha_i)}$ by Theorem \ref{thm7}. Therefore, we complete the proof.
\end{IEEEproof}

\begin{remark}\label{rem2}
The skew CRT of \cite{gao2016chinese} gives a quotient-ring decomposition based on a factorization of $x^n-1$ into pairwise coprime two-sided maximal elements. On the one hand, Theorem~\ref{thm8} has a different role in the present paper. We restrict to $\F_{q^m}[x;x^q,0]$ and choose the moduli $\Lambda_{C(\alpha_i)}(x)$ from prescribed evaluation sets: by Theorem~\ref{thm7}, their product is the minimal skew polynomial of the union of the corresponding conjugacy classes. In addition, these moduli commute as in Lemma~\ref{lem4}, and Lemma~\ref{lem7} restores the product-evaluation identity required by the later conversion formulas. Thus, the feature used in this paper is not merely a quotient decomposition, but the compatibility between the CRT congruences and evaluations on the selected P-bases. 
On the other hand, the following toy Example \ref{exam} shows that the family of moduli arising in Theorem~\ref{thm8} is not contained in the specific $x^n-1$ setting considered in \cite{gao2016chinese}. Specifically, we show there exists $\Lambda_{\cup_{i\in[t]}C(\alpha_i)}(x)\ne x^n-1$ for any positive integer $n$.
\end{remark}

\begin{example}\label{exam}
	Let $q=5$ and $m=2$. Then $\sigma(x)=x^5$ and $\mathrm{N}_{\mathbb{F}_{q^m}/\mathbb{F}_q}=\mathrm{N}_{\mathbb{F}_{5^2}/\mathbb{F}_{5}}$. Choose two distinct conjugacy classes whose representatives have norms $2$ and $3$ in $\mathbb{F}_5^{*}$. By Corollary \ref{cor2}, their minimal skew polynomials are $\Lambda_1(x)=x^2-2$ and $\Lambda_2(x)=x^2-3$. Then, 
	\begin{equation*}
		\Lambda_1(x)\Lambda_2(x) 
		= x^4 - 5x^2 + 6
		= x^4+1 \pmod{5}.
	\end{equation*}
	If there exists $n$ such that $x^4+1=x^n-1$, then $n=4$. This implies a contradiction $1=-1$ over $\mathbb{F}_{5^2}$. Therefore, no such positive integer $n$ exists.
\end{example}

\section{Conversion on linearized Reed-Solomon Codes}\label{sec-conversion LRS}

We now use the evaluation-compatible sCRT of Theorem~\ref{thm8} to build merge conversions at the level of skew evaluation polynomials. The construction is designed so that each initial message polynomial determines one congruence component, while the resulting final polynomial can be evaluated to recover unchanged symbols and generate written symbols from a controlled set of reads. We first review PECs and the skew-polynomial representation of linearized Reed--Solomon codes in Subsection~\ref{sub-pec lrs}, and then treat distinct and identical initial codes separately.

\subsection{polynomial evaluation codes and linearized Reed-Solomon codes}\label{sub-pec lrs}

\begin{definition}[Polynomial evaluation codes (PECs)]\label{def pec}
Let $\mb{f}=(f_i)_{i\in[k]} \in (\F[x;\sigma,\delta])^k$ be a $k$-dimensional vector of polynomials such that $f_1,\dots,f_k$ are linearly independent over $\F$. Let $A=\{\alpha_1,\alpha_2,\dots,\alpha_n\} \subseteq \F$ contain $n$ distinct elements. In general, the polynomial evaluation code $\cC_{\mb{f},A}$ is defined as the linear code with a generator matrix $\mb{G}_{\mb{f},A}$, where
\begin{equation*}
\mb{G}_{\mb{f},A} = \begin{pmatrix}
      f_1(\alpha_1) & f_1(\alpha_2) & \cdots & f_1(\alpha_n) \\
      f_2(\alpha_1) & f_2(\alpha_2) & \cdots & f_2(\alpha_n) \\
      \vdots & \vdots & \ddots & \vdots \\
      f_k(\alpha_1) & f_k(\alpha_2) & \cdots & f_k(\alpha_n)
    \end{pmatrix}.
\end{equation*}
\end{definition}

Before defining the linearized Reed-Solomon codes, we introduce a left vector space isomorphism shown in \cite{martinez-penas2018skew}:
\begin{equation*}
\begin{array}{cccc}
 \F[x;\sigma,\delta] & \to & \F[D_\alpha^{\sigma,\delta}] \\
  f(x)=\sum_{i=0}^{d}a_ix^i & \mapsto & f^{D_\alpha}(x)=\sum_{i=0}^{d}a_iD_\alpha^i(x),
\end{array}
\end{equation*}
where $\F[D_\alpha^{\sigma,\delta}]$ is the left vector space of operator polynomials over $\F$ as that with basis $\{D_\alpha^i:i=0,1,\dots,d\}$.

We will now review the definition of linearized Reed-Solomon codes as follows.

\begin{definition}[Linearized Reed-Solomon codes~\cite{martinez-penas2018skew}]\label{def lrsc}
For each $k=0,1,\dots,n$, define the $k$-dimensional linearized Reed-Solomon code with conjugacy representatives $\alpha_1,\dots,\alpha_\ell\in\F$ and basis vectors $\beta_i=(\beta_{i,1},\dots,\beta_{i,n_i})\in\F^{n_i}$ for $i\in[\ell]$, as the left linear code $\mathcal{C}_{\ell,k}^{\sigma,\delta}\subseteq\F^n$ formed by the vectors $\mathbf{c}=(\mathbf{c}_1,\dots,\mathbf{c}_\ell)\in\F^n$ given by $\mathbf{c}_i=(c_{i,1},\dots,c_{i,n_i})\subseteq\F^{n_i}$ and
\begin{equation*}
    c_{i,j} = \sum_{u=0}^{k-1}a_{u} D_{\alpha_i}^u (\beta_{i,j}),
\end{equation*}
where $a_u\in\F$ for $u\in[k-1]\cup\{0\}, j\in[n_i], i\in[\ell]$.
\end{definition}

As shown in \cite[Thm. 4]{martinez-penas2018skew}, the code $\mathcal{C}_{\ell,k}^{\sigma,\delta}$ defined in Definition \ref{def lrsc} is a maximum sum-rank distance code. It is well known that maximum sum-rank distance codes are MDS codes. Thus, the access cost of conversion on maximum sum-rank distance codes is bounded by Theorem \ref{thm bound}. The following result connects the evaluation of skew polynomials and evaluation of operator polynomials.

\begin{lemma}[\cite{leroy1995pseudo}]\label{lem equiv}
Let $\alpha\in\F,\beta\in\F^*$ and $f(x)\in\F[x;\sigma,\delta]$. Then $f(D_\alpha(\beta)\beta^{-1})=f^{D_\alpha}(\beta)\beta^{-1}$.
\end{lemma}

In fact, Theorem \ref{thm8} can be generalized to the polynomials in $\F[D_\alpha^{\sigma,\delta}]$ since $\F[x;\sigma,\delta]$ and $\F[D_a^{\sigma,\delta}]$ are isomorphic. To directly apply Theorem \ref{thm8}, we consider following skew polynomial form codes which is equivalent to codes from Definition \ref{def lrsc} by Lemmas \ref{lem equiv} and \ref{lem11}, i.e.,
\begin{equation*}
	c_{i,j}^{\mathrm{Def 13}} \beta_{i,j}^{-1} 
	= f^{D_{\alpha_i}}(\beta_{i,j})\beta_{i,j}^{-1}
	= f(D_{\alpha_i}(\beta_{i,j})\beta_{i,j}^{-1})
	= f(\gamma_{i,j})
	= c_{i,j}^{\mathrm{Def 14}}.
\end{equation*}

\begin{definition}\label{def lrsc2}
For each $k=0,1,\dots,n$, define the $k$-dimensional linearized Reed-Solomon code with conjugacy representatives $\alpha_1,\alpha_2,\dots,\alpha_\ell\in\F$ and P-basis $U_i=\{\gamma_{i,1},\gamma_{i,2},\dots,\gamma_{i,n_i}\}$ of conjugacy class $C(\alpha_i)$ for $i\in[\ell]$, as the left linear code
$\mathcal{C}_{\ell,k}^{\sigma,\delta}\subseteq\F^n$ formed by the vectors $\mathbf{c}=(\mathbf{c}_1,\dots,\mathbf{c}_\ell)\in\F^n$ given by $\mathbf{c}_i=(c_{i,1},\dots,c_{i,n_i})\subseteq\F^{n_i}$ and
\begin{equation*}
    c_{i,j} = \sum_{u=0}^{k-1}a_{u} \Phi_u(\gamma_{i,j}),
\end{equation*}
where $a_u\in\F$ for $u\in[k-1]\cup\{0\}, j\in[n_i], i\in[\ell]$.
\end{definition}

In what follows, all $\mathcal{C}_{\ell,k}^{\sigma,\delta}$ are given according to Definition \ref{def lrsc2} instead of Definition \ref{def lrsc}. Observe that a generator matrix for $\mathcal{C}_{\ell,k}^{\sigma,\delta}$ is $\mb{G}_{\mathbf{f},U_1\cup\cdots\cup U_\ell}$, where $\mathbf{f} = (1, x, \dots, x^{k-1}) \in (\F[x;\sigma,\delta])^k$ which is exactly the form of Reed-Solomon codes in the traditional polynomial ring $\F[x]$. 

To ensure the valid application of Theorem \ref{thm8}, we consider the linear code $\mathcal{C}_{\ell,k}^{\sigma,\delta}$ satisfying $k=n_1+n_2+\cdots+n_j$ for some $1\le j\le \ell$ in what follows. Now, we propose constructions of irregular convertible codes on skew polynomials.

\subsection{Construction for distinct initial skew PECs}\label{sub-con distinct lrs}
In this subsection, we propose a construction of irregular convertible codes on distinct initial skew PECs in Construction \ref{con C}. Then, the conversion on linearized Reed-Solomon code is shown in Corollary \ref{cor-con C} derived by Construction \ref{con C}.

\begin{construction}\label{con C}
Let $t,\bar{k},r,m$ are positive integers. Let $\F[x;\sigma,\delta]=\F_{q^m}[x;x^q,0]$ with $q\ge t\bar{k}+r+1$.
\begin{enumerate}
	\item \textbf{Choose P-basis.} Suppose that $U_{i,j}=\{\gamma_{i,j,1}, \gamma_{i,j,2}, \dots, \gamma_{i,j,m}\}$ is a P-basis of conjugacy class $C(\alpha_{i,j})$ for $(i,j)\in[t]\times[\bar{k}]$ and $(i,j)\in\{0\}\times [r]$, where all $\alpha_{i,j}$ are conjugacy representatives.
	\item \textbf{Set Initial Codes.} Let $\mathbf{f} = (f_1, f_2, \dots, f_k) \in (\F_{q^m}[x;x^q,0])^k$ with $k=m\bar{k}$ such that $f_1,f_2,\dots,f_k$ are linearly independent over $\mathbb{F}_{q^m}$. Set initial codes $\mathcal{C}_{\mathbf{f},U_i\cup U_0}$ for $i\in[t]$ with
		\begin{equation*}
			U_i = \bigcup_{j\in[\bar{k}]}U_{i,j},\ U_0 = \bigcup_{j\in[r]}U_{0,j}.
		\end{equation*}
	\item \textbf{Components of sCRT.} Let $U_{-i} = \bigcup_{j\in[t]\setminus\{i\}} U_j$ for $i\in[t]$. Choose any P-basis $\hat{U}_i$ for $i\in[t]$ satisfying
		\begin{equation*}
			U_i\cap \hat{U}_i= U_i,\ U_{-i}\cap \hat{U}_i=\emptyset
		\end{equation*}
		such that
		\begin{equation*}
			\max_{i\in[k]}\{ \deg(f_i(x)) \} \le  \min_{i\in[t]}\{\deg(\Lambda_{\hat{U}_{i}}(x))\} - 1.
		\end{equation*}
	\item \textbf{Conversion.} Let $\hat{U}_{-i}=\bigcup_{j\in[t]\setminus\{i\}}\hat{U}_j$ for $i\in[t]$, and $U = \bigcup_{j\in[z]} U_{0,j} \subseteq U_0$ for some $z\le \min\{\bar{k},r\}$. Define a conversion map
		\begin{equation*}
    		\varphi: \prod_{i=1}^{t}\mathcal{C}_{\mathbf{f},U_i\cup U_0} \to \mathcal{C}_F
		\end{equation*}
		such that $\varphi(c_{\mb{m}_1},\dots,c_{\mb{m}_t}) = (\bar{c}_{\mb{m}_1},\dots,\bar{c}_{\mb{m}_{t}}, \bar{c})$, i.e.,
		\begin{equation*}
    		\cC_F \triangleq  \left\{ \varphi(c_{\mb{m}_1},\dots,c_{\mb{m}_t}) :  c_{\mb{m}_i} \in \cC_{\mb{f},U_i\cup U_0}, i\in[t] \right\},
		\end{equation*}
		where $\bar{c}_{\mb{m}_i}= ((\mb{m}_i \mb{f}^\top)(\gamma))_{\gamma \in U_i}$ are unchanged symbols, and $\bar{c} = (f(\gamma))_{\gamma\in U}$ are written symbols with
		\begin{equation}\label{eqn30}
    		f(x) = \sum_{i=1}^{t} (\mathbf{m}_i\mathbf{f}^\top)(x) \Lambda_{\hat{U}_{-i}}(x).
		\end{equation}
\end{enumerate}
\end{construction}

\begin{remark}
The requirement of $q\ge t\bar{k}+r+1$ in Construction \ref{con C} makes sure there exists $t\bar{k}+r$ distinct conjugacy classes. This comes from the result that, let $\gamma\in\F_{q^m}^*$ be a primitive element of $\F_{q^m}$, then $\gamma^0,\gamma^1,\dots,\gamma^{q-2}$ are representatives of distinct conjugacy classes~\cite{martinez-penas2019Universal}.
\end{remark}

\begin{theorem}\label{thm9}
The map $\varphi$ in Construction \ref{con C} implies a $(t,1)_{q^m}$ irregular convertible code $\sC$ consisting of $t$ initial skew PECs $\cC_{\mb{f},U_1\cup U_0},\cC_{\mb{f},U_2\cup U_0},\dots,\cC_{\mb{f},U_t\cup U_0}$, the final $[tk+|U|,tk]_{q^m}$ skew PEC $\cC_F$. Moreover, the read access cost is $t|U|$ and the write access cost is $|U|$.
\end{theorem}

\begin{IEEEproof}
Denote $C = \bigcup_{i\in[t],j\in[\bar{k}]} C(\alpha_{i,j})$. Obviously, $f(x)$ in \eqref{eqn30} satisfies
\begin{equation}\label{eqn31}
    f(x) \equiv (\mathbf{m}_i\mathbf{f}^\top)(x) \Lambda_{\hat{U}_{-i}}(x) \pmod{\Lambda_{\hat{U}_i}(x)} \mbox{ for each } i\in[t],
\end{equation}
and, by Theorem \ref{thm7},
\begin{equation*}
    \Lambda_{\hat{U}_{-i}} = \prod_{j\in[t]\setminus\{i\}} \Lambda_{\hat{U}_j}.
\end{equation*}
Note that
\begin{equation}\label{eqn-degf}
\begin{aligned}
    \deg(f)
    &\le \max_{i\in[t]} \left\{ \max_{j\in[k]}\{\deg(f_j)\} + \deg(\Lambda_{\hat{U}_{-i}}) \right\} \\
    &\le \max_{i\in[t]} \{ \deg(\Lambda_{\hat{U}_i})-1 + \deg(\Lambda_{\hat{U}_{-i}}) \} \\
    &= \deg(\Lambda_{\cup_{i\in[t]}\hat{U}_i})-1.
\end{aligned}
\end{equation}
According to Theorem \ref{thm8}, $f(x)$ is the unique skew polynomial satisfying the system of congruence equations \eqref{eqn31} for given $\mb{m}_1,\mb{m}_2,\dots,\mb{m}_t$. Consider the skew PEC as
\begin{equation}\label{eqn-equive code}
\tilde{\cC} \triangleq \left\{ (f(\gamma))_{\gamma\in U_1\cup \cdots\cup U_t\cup U} : \mb{m}_i\in\F_{q^m}^k, i\in[t] \right\}.
\end{equation}
According to the uniqueness of $f(x)$, the code $\tilde{\mathcal{C}}$ is a PEC with dimension $\log_{q^m}|\mathbb{F}_{q^m}^k|^t=tk$ and length $tk+|U|$. By Lemma \ref{lem4},
\begin{equation}\label{eqn-scalar}
    f(\gamma) = (\mb{m}_i \mb{f}^\top)(\gamma) \Lambda_{\hat{U}_{-i}}(\gamma),
\end{equation}
where $\Lambda_{\hat{U}_{-i}}(\gamma) \ne 0$ for all $\gamma\in U_i\subseteq \hat{U}_i$. Hence, the final code $\cC_F$ is equivalent to the code $\tilde{\cC}$. Therefore, $\cC_F$ is a PEC with dimension $tk$ and length $tk+|U|$.

According to \eqref{eqn30}, we have that
\begin{equation*}
f(\gamma) = \sum_{i=1}^{t} (\mathbf{m}_i\mathbf{f}^\top)(\gamma) \Lambda_{\hat{U}_{-i}}(\gamma) \mbox{ for all } \gamma\in U.
\end{equation*}
This is to say that, we can generate the written symbols $\bar{c}$ by reading $(\mb{m}_i \mb{f}^\top)(\gamma)$ for all $\gamma\in U, i\in[t]$, i.e., the read access cost is $t|U|$ and write access cost is $|U|$, since $\Lambda_{\hat{U}_{-i}}(\gamma)$ can be regarded as constant independent with $\mathbf{m}_j$ for $j\in [t]$.
\end{IEEEproof}

\begin{remark}
The P-basis in Construction \ref{con C} can be derived from Theorem \ref{thm7}, i.e., let $\gamma_{i,j,\ell}=\beta_{\ell}^{q-1}\alpha_{i,j}$ for $(i,j)\in[t]\times[\bar{k}]\cup \{0\}\times[r]$ and $\ell\in[m]$, where $\{\beta_1,\beta_2,\dots,\beta_m\}$ is a basis of $\F_{q^m}$ over $\F_q$.
\end{remark}

The following corollary is direct from Theorem \ref{thm9} and Definition \ref{def lrsc2}.

\begin{corollary}\label{cor-con C}
	Let $\mathbf{f} = (1, x, \dots, x^{k-1}) \in (\F_{q^m}[x;x^q,0])^k$ and $\hat{U}_i=U_i$ for $i\in[t]$ in Construction \ref{con C}. Then, the irregular convertible code $\sC$ converts linearized Reed-Solomon codes $\mathcal{C}_{\mathbf{f},U_i\cup U_0} = \mathcal{C}_{\bar{k}+r,k}^{x^q,0}$ for $i\in[t]$ into a final code $\mathcal{C}_F$ equivalent to linearized Reed-Solomon code $\mathcal{C}_{\hat{\mathbf{f}}=(1,x,\dots,x^{tk-1}),U_1\cup \cdots\cup U_t\cup U} = \mathcal{C}_{t\bar{k}+z,tk}^{x^q,0}$ with per-symbol access-optimal cost.
\end{corollary}

\begin{IEEEproof}
	In this case, $f(x)$ in \eqref{eqn31} satisfies $\deg(f)\le tk-1$ by \eqref{eqn-degf} and $\deg(\Lambda_{\cup_{i\in[t]}\hat{U}_i})=|\cup_{i\in[t]}\hat{U}_i|=|\cup_{i\in[t]}U_i|=tk$. Thus, $f(x)$ can be formed as $f(x)=\sum_{i=0}^{tk-1}a_ix^i\in\mathbb{F}_{q^m}[x;x^q,0]$. Therefore, the final code is equivalent to a linearized Reed-Solomon code by Definition \ref{def lrsc2} and $\tilde{\mathcal{C}}$ in \eqref{eqn-equive code}. Since any linearized Reed-Solomon code is an MDS code, together with \eqref{eqn30} and Theorem \ref{thm9}, we have $\sC$ is per-symbol access-optimal.
\end{IEEEproof}

\subsection{Construction for identical initial skew PECs}\label{sub-con identical lrs}

Construction~\ref{con C} assigns different evaluation sets to the distinct initial codes, and these sets naturally provide different minimal skew polynomials for the sCRT. The identical-code case is more delicate: all initial codewords come from the same evaluation code, so the original evaluation set alone does not provide $t$ separate CRT components. Auxiliary P-bases must therefore be introduced to create the required moduli, while the final code must still preserve the designated unchanged symbols and compute every written symbol from the allowed read symbols.

This is the purpose of the two compatibility conditions in Construction~\ref{con D}. Condition \textbf{C1} ensures that the evaluation map on each $U_i$ is invertible. Hence, the $k$ unchanged-symbol values originally associated with the common evaluation set $U_1$ can be re-expressed as evaluations of an auxiliary polynomial on the distinct set $U_i$. Condition \textbf{C2} is a read-symbol compatibility condition: after transporting the representation associated with the unchanged-symbol set, the evaluations required at $U_{0,0}$ must lie in the span of the evaluations available from $U_{0,i}$. These conditions make explicit the additional structure needed in the identical-code case, rather than treating it as a direct copy of Construction~\ref{con C}.

We first state Construction~\ref{con D}, then give a sufficient structural condition for \textbf{C2} in Lemma~\ref{thm11}, and finally specialize the construction to linearized Reed--Solomon codes in Corollary~\ref{cor5}.

\begin{construction}\label{con D}
Let $t,m,\ell,\ell'$ are positive integers. Let $\F[x;\sigma,\delta]=\F_{q^m}[x;x^q,0]$ with $q\ge t\ell+\ell'+1$.
\begin{enumerate}
	\item \textbf{Choose P-basis.}
		\begin{enumerate}
			\item \textbf{Unchanged Symbols.} Let $U_i=\{ \gamma_{i,1},\gamma_{i,2},\dots,\gamma_{i,k} \}$ be P-basis of $C_i = \bigcup_{j\in[\ell]} C(\alpha_{i,j})$ for $i\in[t]$, where $\alpha_{i,j}, i\in[t], j\in[\ell]$ are conjugacy representatives. Obviously, $k=m\ell$.
			\item \textbf{Read Symbols.} Let $U_0=\{ \gamma_{0,1},\gamma_{0,2},\dots,\gamma_{0,k'} \}$ be P-basis of $C_0 = \bigcup_{j\in[\ell']} C(\alpha_{0,j})$, where $\alpha_{0,j}, j\in[\ell']$ are conjugacy representatives. Obviously, $k'=m\ell'$. Let $U_{0,i}=\{ \gamma_{0,i,1},\gamma_{0,i,2},\dots,\gamma_{0,i,r} \} \subseteq U_0$  be P-basis of  $\bigcup_{j\in[z]}C(\zeta_{i,j})$ for $i\in[t]\cup\{0\}$ and some $z\le\min\{\ell,\ell'\}$, where $\{\zeta_{i,j}\}_{j\in[z]} \subseteq \{\alpha_{0,j}\}_{j\in[\ell']}$. Obviously, $r=mz$.
		\end{enumerate} 
	\item \textbf{Set Initial Codes.} Let $\mb{f}=(f_1,f_2,\dots,f_{k}) \in (\F_{q^m}[x;x^q,0])^{k}$ such that $f_1,f_2,\dots,f_{k}$ are linearly independent over $\mathbb{F}_{q^m}$. Set the initial code $\mathcal{C}_{\mathbf{f},U_1\cup U_0}$. 
	\item \textbf{Components of sCRT.} Let $U_{-i} = \bigcup_{j\in[t]\setminus\{i\}} U_j$ for $i\in[t]$. Choose any P-basis $\hat{U}_i$ for $i\in[t]$ satisfying
		\begin{equation*}
			U_i\cap \hat{U}_i= U_i,\ U_{-i}\cap \hat{U}_i=\emptyset
		\end{equation*}
		such that 
		\begin{equation*}
			\max_{i\in[k]}\{ \deg(f_i(x)) \} \le  \min_{i\in[t]}\{\deg(\Lambda_{\hat{U}_{i}}(x))\} - 1.
		\end{equation*}
	\item \textbf{Restricted Structures.} Suppose the following conditions hold:
		\begin{enumerate}
  			\item[\textbf{C1}:] $\mb{G}_{\mb{f},U_i}$ is an invertible matrix for $i\in[t]$ defined in Definition \ref{def pec}.
  			\item[\textbf{C2}:] For $i\in[t]$, there exists an invertible diagonal matrix $\mb{D}_i=\diag(d_{i,1},d_{i,2},\dots,d_{i,k}) \in \F_{q^m}^{k\times k}$ such that
      				\begin{equation*}
        			\anglenv{ \mb{G}_{\mb{f},U_1} \mb{D}_i \mb{G}_{\mb{f},U_i}^{-1} \mb{G}_{\mb{f},U_{0,0}} } \subseteq \anglenv{\mb{G}_{\mb{f},U_{0,i}}}
      				\end{equation*}
      			i.e.,
      				\begin{equation*}
      				\mb{G}_{\mb{f},U_1} \mb{D}_i \mb{G}_{\mb{f},U_i}^{-1} \mb{G}_{\mb{f},U_{0,0}}
      				= \mb{G}_{\mb{f},U_{0,i}} \mb{L}_i,
      				\end{equation*}
     			where $\anglenv{\mb{G}_{\mb{f},U_{0,i}}}$ denotes the linear space spanned by the columns of $\mb{G}_{\mb{f},U_{0,i}}$ and  $\mb{L}_i = (\ell_{i,z,j})_{z,j\in[r]}$ is an $r\times r$ matrix for $i\in[t]$.
            This condition ensures that the evaluations needed to form the written symbols can be reconstructed from the designated read-symbol evaluations on $U_{0,i}$.
		\end{enumerate}
	\item \textbf{Conversion.} Let $\hat{U}_{-i}=\bigcup_{j\in[t]\setminus\{i\}}\hat{U}_j$ for $i\in[t]$. Define a conversion map $\varphi : \prod_{i=1}^{t} \cC_{\mb{f},U_1\cup U_0} \to \cC_F$ such that $\varphi(c_{\mb{m}_1},\dots,c_{\mb{m}_t}) = (\bar{c}_{\mb{m}_1},\dots,\bar{c}_{\mb{m}_{t}}, \bar{c})$, i.e.,
		\begin{equation*}
			\cC_F \triangleq  \left\{ \varphi(c_{\mb{m}_1},\dots,c_{\mb{m}_t}) :  c_{\mb{m}_i} \in \cC_{\mb{f},U_1\cup U_0}, i\in[t] \right\},
		\end{equation*}
		where $\bar{c}_{\mb{m}_i} = ((\mb{m}_i \mb{f}^\top)(\gamma))_{\gamma\in U_1}$ are unchanged symbols for $i\in[t]$, and $\bar{c} = (y_j(\gamma_{0,0,j}))_{j\in[r]}$ are written symbols with
			\begin{equation*}
 				y_j(\gamma_{0,0,j}) = \sum_{i=1}^{t} \parenv{ \sum_{z=1}^{r} \ell_{i,z,j} (\mb{m}_i \mb{f}^\top)(\gamma_{0,i,z}) }  \Lambda_{\hat{U}_{-i}}(\gamma_{0,0,j}).
			\end{equation*}
\end{enumerate}
\end{construction}

\begin{theorem}\label{thm10}
The map $\varphi$ in Construction \ref{con D} implies a $(t,1)_{q^m}$ irregular convertible code $\sC$ consisting of the initial skew PEC $\cC_{\mb{f},U_1 \cup U_0}$, the final $[tk+r,tk]_{q^m}$ skew PEC $\cC_F$. Moreover, the read access cost is $tr$ and the write access cost is $r$.
\end{theorem}

\begin{IEEEproof}
Recall that $\cC_{\mb{f},U_1\cup U_0} = \{
((\mb{m} \mb{f}^\top)(\gamma))_{\gamma\in U_1\cup U_0} : \mb{m} \in \F_{q^m}^{k}\}.$ Choose any codewords $\mb{c} = (c_{\mb{m}_1},c_{\mb{m}_2},\dots,c_{\mb{m}_t}) \in (\cC_{\mb{f},U_1\cup U_0})^t$ such that $\mb{m}_i \in \F_{q^m}^{k}$ for $i\in[t]$. We consider a skew polynomial
\begin{equation*}
e_i(x) = (a_1,\dots,a_{k}) \cdot \mb{f}^\top \in \F_{q^m}[x;x^q,0]
\end{equation*}
such that
\begin{equation*}
\coff(e_i)
=
\coff(\mb{m}_i \mb{f}^\top) \mb{G}_{\mb{f},U_1} \mb{D}_i \mb{G}_{\mb{f},U_i}^{-1},
\end{equation*}
where $\coff(e_i) \triangleq (a_1,a_2,\dots,a_{k})$ for polynomial $e_i(x)$, the notation $\coff(\mb{m}_i \mb{f}^\top)$ is defined analogously. Therefore,
\begin{equation*}
    \coff(e_i) \mb{G}_{\mb{f},{ U_i}}
    =\coff(\mb{m}_i \mb{f}^\top) \mb{G}_{\mb{f},U_1} \mb{D}_i,
\end{equation*}
which implies $e_i(x)$ satisfies
\begin{equation}\label{eqn32}
e_i(\gamma_{i,j}) = d_{i,j} (\mb{m}_i \mb{f}^\top)(\gamma_{1,j}) \mbox{ for all } j\in[k].
\end{equation}
Based on the condition \textbf{C2}, we have that
\begin{equation*}
    \coff(e_i) \mb{G}_{\mb{f},U_{0,0}}
    = \coff(\mb{m}_i \mb{f}^\top) \mb{G}_{\mb{f},U_1} \mb{D}_i \mb{G}_{\mb{f},U_i}^{-1} \mb{G}_{\mb{f},U_{0,0}}
    = \coff(\mb{m}_i  \mb{f}^\top) \mb{G}_{\mb{f},U_{0,i}} \mb{L}_i.
\end{equation*}
This is to say
\begin{equation}\label{eqn33}
e_i(\gamma_{0,0,j}) = \sum_{z\in[r]} \ell_{i,z,j} (\mb{m}_i \mb{f}^\top)(\gamma_{0,i,z}) \mbox{ for all } j\in[r].
\end{equation}
Obviously,
\begin{equation}\label{eqn37}
f(x) = \sum_{i=1}^{t} e_i(x) \Lambda_{\hat{U}_{-i}}(x)
\end{equation}
satisfies
\begin{equation}\label{eqn34}
f(x) \equiv e_i(x) \Lambda_{\hat{U}_{-i}}(x) \pmod{\Lambda_{\hat{U}_i}(x)} \mbox{ for all } i\in[t],
\end{equation}
and
\begin{equation}\label{eqn-deg f2}
\begin{split}
\deg(f)
&\le \max_{i\in[t]} \left\{ \deg(e_i) + \deg(\Lambda_{\hat{U}_{-i}}) \right\} \\
&\le \max_{i\in[t]} \left\{ \max_{i\in[k]}\{ \deg(f_i(x)) \} + \deg(\Lambda_{\hat{U}_{-i}}) \right\} \\
&\le \max_{i\in[t]} \left\{ \deg(\Lambda_{\hat{U}_i})-1 + \deg(\Lambda_{\hat{U}_{-i}}) \right\} \\
&= \deg(\Lambda_{\cup_{i\in[t]}\hat{U}_i})-1.
\end{split}
\end{equation}
Based on Theorem \ref{thm8}, $f(x)$ is the unique skew polynomial satisfying the system of congruence equations \eqref{eqn34} for given $\mb{m}_1,\mb{m}_2,\dots,\mb{m}_t$. Consider the PEC as
\begin{equation}\label{eqn-equive code2}
\tilde{\cC} \triangleq \left\{ (f(\gamma))_{ \gamma\in U_1\cup \cdots\cup U_t\cup U_{0,0} } : \mb{m}_i\in\F_{q^m}^{k}, i\in[t] \right\}.
\end{equation}
By the uniqueness of $f(x)$, the dimension of $\tilde{\mathcal{C}}$ is $\log_{q^m}|\F_{q^m}^{k}|^t=tk$ and length $tk+r$.

By Lemma \ref{lem7} and equations \eqref{eqn32}, \eqref{eqn33}, we have that
\begin{equation}\label{eqn35}
    f(\gamma_{i,j})
    = e_i(\gamma_{i,j}) \Lambda_{\hat{U}_{-i}}(\gamma_{i,j})
    = d_{i,j} (\mb{m}_i \mb{f}^\top)(\gamma_{1,j}) \Lambda_{\hat{U}_{-i}}(\gamma_{i,j})
\end{equation}
for each $\gamma_{i,j} \in U_i\subseteq \hat{U}_i, i\in[t]$, and
\begin{equation}\label{eqn36}
\begin{aligned}
    f(\gamma_{0,0,j})
    &= \sum_{i=1}^{t} e_i(\gamma_{0,0,j}) \Lambda_{\hat{U}_{-i}}(\gamma_{0,0,j}) \\
    &= \sum_{i=1}^{t} \parenv{ \sum_{z=1}^{r} \ell_{i,z,j} (\mb{m}_i \mb{f}^\top)(\gamma_{0,i,z}) } \Lambda_{\hat{U}_{-i}}(\gamma_{0,0,j}) \\
    &= y_j(\gamma_{0,0,j})
\end{aligned}
\end{equation}
for each $\gamma_{0,0,j}\in U_{0,0}$. According to Lemma \ref{lem4}-2) and $\mb{D}_i$ is an invertible matrix for $i\in[t]$, we have $d_{i,j}\Lambda_{\hat{U}_{-i}}(\gamma_{i,j})\ne 0$. Therefore, \eqref{eqn35} and \eqref{eqn36} imply that the code $\tilde{\cC}$ is equivalent to the code $\cC_F$. Thus, $\cC_F$ is a skew PEC with dimension $tk$ and length $tk+r$.

According to \eqref{eqn36}, each $y_j(\gamma_{0,0,j})$ for $j\in[r]$ can be generated by reading $(\mb{m}_i \mb{f}^\top)(\gamma_{0,i,z})$ for $i\in[t]$ and $z\in[r]$. Thus, the read access cost is $tr$, the write access cost is $r$.
\end{IEEEproof}

\begin{remark}\label{rem-f}
    In fact, the polynomials $f(x)$ in \eqref{eqn30} and  \eqref{eqn37} can be chosen as the form in \eqref{eqn41}. However, the forms in \eqref{eqn30} and  \eqref{eqn37} are more concise since these avoid computing polynomials $N_i(x)$ by Lagrange interpolation. It is worth noting that, if form \eqref{eqn41} is chosen, then the final code $\mathcal{C}_F$ is exactly $\tilde{\mathcal{C}}$ shown in \eqref{eqn-equive code} and \eqref{eqn-equive code2} in Constructions \ref{con C} and \ref{con D}, respectively.
\end{remark}

In what follows, we propose a sufficient condition for \textbf{C2} in Construction \ref{con D} to hold in Lemma \ref{thm11} and show it is easy to satisfy in Corollary \ref{cor5} to construct convertible codes for linearized Reed-Solomon codes.

\begin{lemma}\label{thm11}
    Use the notation in Construction \ref{con D}. For $i\in[t]$, if the conjugacy representatives satisfy
    \begin{equation}\label{eqn_delta_i}
    \begin{split}
    	\{\alpha_{i,j}\}_{j\in[\ell]} &= \delta_i \{\alpha_{1,j}\}_{j\in[\ell]}, \\
        \{\zeta_{i,j}\}_{j\in[z]} &= \delta_i^{-1} \{\zeta_{0,j}\}_{j\in[z]},
    \end{split}
    \end{equation}
    and there exists invertible matrices $\mathbf{Q}_1,\mathbf{Q}_2,\dots,\mathbf{Q}_t$ such that 
    \begin{equation}\label{eqn-QG}
    \begin{split}
    	\mathbf{Q}_i  \mathbf{G}_{\mathbf{f},U_{1}} &= \mathbf{G}_{\mathbf{f},\delta_i U_{1}}, \\
    	\mathbf{Q}_i \mathbf{G}_{\mathbf{f},U_{0,i}} &= \mathbf{G}_{\mathbf{f},\delta_i U_{0,i}}.
    \end{split}
    \end{equation} 
    Then, $\mathbf{D}_i=\mathbf{I}_k$ and $\mathbf{L}_i=\mathbf{I}_r$ for $i\in[t]$ can be chosen in Construction \ref{con D}.
\end{lemma}

\begin{IEEEproof}
	According to Corollary \ref{cor3}, each $\gamma_{i}$ in P-basis $U=\{\gamma_1,\dots,\gamma_m\}$ of $C(\alpha)$ is exactly $\gamma_{i}=\beta_i^{q-1}\alpha$, where $\{\beta_1,\dots,\beta_m\}$ is a basis of $\mathbb{F}_{q^m}$ over $\mathbb{F}_q$ and $\alpha$ is a conjugacy representative. By \eqref{eqn_delta_i}, $\mathbf{G}_{\mathbf{f},\delta_i U_{1}} = \mathbf{G}_{\mathbf{f},U_{i}}$ and $\mathbf{G}_{\mathbf{f},\delta_i U_{0,i}} = \mathbf{G}_{\mathbf{f},U_{0,0}}$ for $i\in[t]$. 
	Then, by \eqref{eqn-QG}, we have
	\begin{align*}
		\mb{G}_{\mb{f},U_1} \mb{G}_{\mb{f},U_i}^{-1} \mb{G}_{\mb{f},U_{0,0}}
		&= \mb{G}_{\mb{f},U_1} (\mathbf{Q}_{i}\mb{G}_{\mb{f},U_1})^{-1} (\mathbf{Q}_i \mathbf{G}_{\mathbf{f},U_{0,i}}) \\
		&= \mathbf{G}_{\mathbf{f},U_{0,i}}. 
	\end{align*}
	Therefore, the condition \textbf{C2} in Construction \ref{con D} holds. Moreover, $\mathbf{D}_i=\mathbf{I}_k$ and $\mathbf{L}_i=\mathbf{I}_r$ for $i\in[t]$.
\end{IEEEproof}

The following lemma ensures the condition \textbf{C1} holds in Construction \ref{con D} for linearized Reed-Solomon codes, i.e., set $\mathbf{f} = (1, x, \dots, x^{k-1}) \in (\F_{q^m}[x;x^q,0])^k$.

\begin{lemma}[{\cite[Lem. 9]{gopi2022improved}}]\label{lem13}
Let $U\subseteq\F$ of size $k$. Let $U=U_1\cup U_2\cup \cdots\cup U_r$ be the partition of $U$ into different conjugacy classes, where $U_i=\{ \gamma_{i,j}= D_{\alpha_i}(\beta_{i,j}) \beta_{i,j}^{-1} : j\in[|U_i|] \}$. Then,
\begin{equation*}
    \begin{pmatrix}
      \Phi_0(\gamma_{1,1}) & \Phi_0(\gamma_{1,2}) & \cdots & \Phi_0(\gamma_{r,|U_r|}) \\
      \Phi_1(\gamma_{1,1}) & \Phi_1(\gamma_{1,2}) & \cdots & \Phi_1(\gamma_{r,|U_r|}) \\
      \vdots & \vdots &  & \vdots \\
      \Phi_{k-1}(\gamma_{1,1}) & \Phi_{k-1}(\gamma_{1,2}) & \cdots & \Phi_{k-1}(\gamma_{r,|U_r|})
    \end{pmatrix}
\end{equation*}
is full rank if for each $i\in[r]$, $\{\beta_{i,j} : j\in[|U_i|]\}$ are linearly independent over the centralizer subfield $\cK_{\alpha_i}$.
\end{lemma}

\begin{corollary}\label{cor5}
    Let $\mathbf{f} = (1, x, \dots, x^{k-1}) \in (\F_{q^m}[x;x^q,0])^k$ and $\hat{U}_i=U_i$ for $i\in[t]$ in Construction \ref{con D}. Set $\alpha_{i,j}=\theta^{j+(i-1)\ell}$ for $i\in[t],j\in[\ell]$, and $\zeta_{i,j}=\theta^{j+(i-1)\ell+t\ell}$ for $i\in[t], j\in[z]$, and $\{\zeta_{0,j}\}_{j\in[z]} = \{\zeta_{1,j}\}_{j\in[z]}$. Then, $\mathbf{D}_i=\mathbf{I}_k$ and $\mathbf{L}_i=\mathbf{I}_r$ for $i\in[t]$. Moreover, the irregular convertible code $\sC$ converts a linearized Reed-Solomon code $\mathcal{C}_{\mathbf{f},U_1\cup U_0} = \mathcal{C}_{\ell+\ell',k}^{x^q,0}$ into a final code $\mathcal{C}_F$ equivalent to a linearized Reed-Solomon code  $\mathcal{C}_{\hat{\mathbf{f}}=(1,x,\dots,x^{tk-1}),U_1\cup \cdots\cup U_t\cup U_{0,0}} = \mathcal{C}_{t\ell+z,tk}^{x^q,0}$ with per-symbol access-optimal.
\end{corollary}

\begin{IEEEproof}
	Let $\delta_i=\theta^{(i-1)\ell}$ and $\mathbf{Q}_{i}=\diag(\Phi_0(\delta_i),\dots,\Phi_{k-1}(\delta_i))$ for $i\in[t]$. By Lemma \ref{lem13}, $\Phi_i(\alpha)=\alpha^{1+q+q^2+\cdots+q^{i-1}}$ and $\Phi_i(\alpha\beta)=\Phi_i(\alpha)\Phi_i(\beta)$ for $i\ge1$, we have
    \begin{equation*}
    \begin{split}
    	\diag(\Phi_0(\delta_i),\dots,\Phi_{k-1}(\delta_i)) \mathbf{G}_{\mathbf{f},U_{1}} &= \mathbf{G}_{\mathbf{f},\delta_i U_{1}} = \mathbf{G}_{\mathbf{f},U_{i}}, \\
        \diag(\Phi_0(\delta_i),\dots,\Phi_{k-1}(\delta_i)) \mathbf{G}_{\mathbf{f},U_{0,i}} &= \mathbf{G}_{\mathbf{f},\delta_i U_{0,i}} = \mathbf{G}_{\mathbf{f},U_{0,0}}.
    \end{split}
    \end{equation*}
    Therefore, Lemma \ref{thm11} holds. This leads to $\mathbf{D}_i=\mathbf{I}_k$ and $\mathbf{L}_i=\mathbf{I}_r$ for $i\in[t]$.
    
    In this case, $f(x)$ in \eqref{eqn37} satisfies $\deg(f)\le tk-1$ by \eqref{eqn-deg f2} and $\deg(\Lambda_{\cup_{i\in[t]}\hat{U}_i})=|\cup_{i\in[t]}\hat{U}_i|=|\cup_{i\in[t]}U_i|=tk$. Thus, $f(x)$ can be formed as $f(x)=\sum_{i=0}^{tk-1}a_ix^i\in\mathbb{F}_{q^m}[x;x^q,0]$. Therefore, the final code is equivalent to a linearized Reed-Solomon code by Definition \ref{def lrsc2} and $\tilde{\mathcal{C}}$ in \eqref{eqn-equive code2}. Since any linearized Reed-Solomon code is an MDS code, together with \eqref{eqn36}, $\mathbf{L}_i=\mathbf{I}_r$ and Theorem \ref{thm10}, we have $\sC$ is per-symbol access-optimal. 
\end{IEEEproof}


\begin{remark}
Constructions~\ref{con C} and~\ref{con D} currently require $m\mid k$. This restriction is not merely notational: the construction groups evaluation points by complete P-bases of conjugacy classes so that the associated minimal skew polynomials have the commutativity property used in Lemmas~\ref{lem4} and~\ref{lem7}. If $k=|U_1|+\cdots+|U_{j-1}|+\theta$ with $0<\theta<|U_j|=m$, the minimal skew polynomial of the partial P-basis $U_j'=\{\gamma_{j,1},\ldots,\gamma_{j,\theta}\}$ has degree $\theta<m$ by Lemma~\ref{lem10}, and the commutativity argument used for complete conjugacy classes is no longer available in the present proof. Relaxing this divisibility restriction would therefore require additional algebraic ideas rather than a direct parameter change.
\end{remark}

\section{Commutative Specialization to $\mathbb{F}_q[x]$}\label{sec-degraded}

The skew framework has a useful commutative specialization. Over $\mathbb{F}_q$, the Frobenius map satisfies $\alpha^q=\alpha$, so $\mathbb{F}_q[x;x^q,0]=\mathbb{F}_q[x]$ and the conjugacy classes reduce to singletons. Consequently, the skew minimal-polynomial and sCRT machinery reduces to the ordinary annihilating-polynomial/CRT picture. This specialization is important for two reasons: it shows that the skew construction is consistent with familiar polynomial evaluation codes, and it allows the same conversion principle to be compared directly with known Reed--Solomon and Tamo--Barg constructions and then applied to the Gabidulin evaluation basis. We therefore derive the commutative counterparts of Constructions~\ref{con C} and~\ref{con D} in Subsections~\ref{sub-de-conc} and~\ref{sub-de-cond}.

\subsection{Degraded conclusions in $\mathbb{F}_q[x]=\mathbb{F}_q[x;x^q,0]$}

In this subsection, we only briefly show a degraded version of the conclusions in Section \ref{sec-CRT}.

\begin{enumerate}
	\item The commutativity is trivial.
	\item The power function satisfies $\Phi_m(\alpha)=\alpha^{1+q+\cdots+q^{m-1}}=\alpha^m$, where $f(x)=x^m\in\mathbb{F}_q[x;x^q,0]$.
	\item The conjugacy class of $\alpha$ is $C(\alpha)=\{D_{\alpha}^{x^q,0}(\beta)\beta^{-1}:\beta\in\mathbb{F}_q^*\} = \{\beta^{q-1}\alpha:\beta\in\mathbb{F}_q^*\} = \{\alpha\}$.
	\item The annihilating skew polynomial for the conjugacy class $C(\alpha)$ is exactly monomial $f(x)=x-\alpha\in\mathbb{F}_q[x]$ which also is the minimal skew polynomial since the P-basis of $C(\alpha)$ is exactly $U=\{\alpha\}$ by Corollary \ref{cor3}.
	\item Theorem \ref{thm7} is degraded to traditional annihilating polynomial $\Lambda_A(x)$ over point set $A$ as $\Lambda_A(x)=\prod_{a\in A}(x-a)\in\mathbb{F}_q[x]$.
\end{enumerate}

Therefore, the evaluation-compatible sCRT in Theorem~\ref{thm8} specializes to the standard CRT in $\mathbb{F}_q[x]$ as follows.

\begin{theorem}[The Chinese Remainder Theorem (CRT)~\cite{ding1996Chinese}]\label{thm CRT}
Let $h_1(x),\dots,h_t(x) \in \F_q[x]$ be pairwise co-prime polynomials. Then for any $t$ polynomials $m_1(x), m_2(x), \dots, m_t(x) \in \F_q[x]$, there exists a unique polynomial $f(x) \in \F_q[x]$ of degree less than $\sum_{i=1}^{t} \deg(h_i(x))$ such that
\begin{equation*}
f(x) \equiv m_i(x) \pmod{h_i(x)}\quad \text{for all } i\in[t].
\end{equation*}
\end{theorem}

In what follows, we can naturally give the degenerate form of Constructions \ref{con C} and \ref{con D}. Here is only one thing that needs to be specially explained. Unlike ordinary polynomials, a skew polynomial need not have at most $\deg(f)$ roots in the usual sense; therefore, we did not state a Hamming-distance bound for Constructions \ref{con C} and \ref{con D}. In the degenerate case, we supplemented the explanation of the minimum distance.

\subsection{Degraded Construction \ref{con C}}\label{sub-de-conc}

In this section, we propose degraded Construction \ref{con C} and also give corresponding conclusions.

\begin{construction}[Degraded Construction \ref{con C}]\label{con A}
Let $t,\bar{k},r$ are positive integers. Let $\F[x;\sigma,\delta]=\F_{q}[x;x^q,0]=\mathbb{F}_q[x]$ with $q\ge t\bar{k}+r+1$.
\begin{enumerate}
	\item \textbf{Choose P-basis.} Let $U_{i,j}=\{\alpha_{i,j}\}$ for $(i,j)\in[t]\times[\bar{k}]$ and $(i,j)\in\{0\}\times [r]$, where all $\alpha_{i,j}$ are pairwise distinct.
	\item \textbf{Set Initial Codes.} Let $\mathbf{f} = (f_1, f_2, \dots, f_k) \in (\F_{q}[x])^k$ with $k=\bar{k}$ such that $f_1,f_2,\dots,f_k$ are linearly independent over $\mathbb{F}_{q}$. Set initial codes $\mathcal{C}_{\mathbf{f},U_i\cup U_0}$ for $i\in[t]$ with
		\begin{equation*}
			U_i = \bigcup_{j\in[\bar{k}]}U_{i,j},\ U_0 = \bigcup_{j\in[r]}U_{0,j}.
		\end{equation*}
	\item \textbf{Components of sCRT.} Let $U_{-i} = \bigcup_{j\in[t]\setminus\{i\}} U_j$ for $i\in[t]$. Choose any set $\hat{U}_i\subseteq \mathbb{F}_q$ for $i\in[t]$ satisfying
		\begin{equation*}
			U_i\cap \hat{U}_i= U_i,\ U_{-i}\cap \hat{U}_i=\emptyset
		\end{equation*}
		such that
		\begin{equation*}
			\max_{i\in[k]}\{ \deg(f_i(x)) \} \le  \min_{i\in[t]}\{\deg(\Lambda_{\hat{U}_{i}}(x))\} - 1.
		\end{equation*}
	\item \textbf{Conversion.} Let $\hat{U}_{-i}=\bigcup_{j\in[t]\setminus\{i\}}\hat{U}_j$ for $i\in[t]$, and $U = \bigcup_{j\in[z]} U_{0,j} \subseteq U_0$ for some $z\le\min\{\bar{k}, r\}$. Define a conversion map
		\begin{equation*}
    		\varphi: \prod_{i=1}^{t}\mathcal{C}_{\mathbf{f},U_i\cup U_0} \to \mathcal{C}_F
		\end{equation*}
		such that $\varphi(c_{\mb{m}_1},\dots,c_{\mb{m}_t}) = (\bar{c}_{\mb{m}_1},\dots,\bar{c}_{\mb{m}_{t}}, \bar{c})$, i.e.,
		\begin{equation*}
    		\cC_F \triangleq  \left\{ \varphi(c_{\mb{m}_1},\dots,c_{\mb{m}_t}) :  c_{\mb{m}_i} \in \cC_{\mb{f},U_i\cup U_0}, i\in[t] \right\},
		\end{equation*}
		where $\bar{c}_{\mb{m}_i}= ((\mb{m}_i \mb{f}^\top)(\gamma))_{\gamma \in U_i}$ are unchanged symbols, and $\bar{c} = (f(\gamma))_{\gamma\in U}$ are written symbols with
		\begin{equation}\label{eqn30-comm}
    		f(x) = \sum_{i=1}^{t} (\mathbf{m}_i\mathbf{f}^\top)(x) \Lambda_{\hat{U}_{-i}}(x).
		\end{equation}
\end{enumerate}
\end{construction}

\begin{theorem}\label{thm3}
The map $\varphi$ in Construction \ref{con A} implies a $(t,1)_q$ irregular convertible code $\sC$ consisting of $t$ initial codes $\cC_{\mb{f},U_1\cup U_0},\cC_{\mb{f},U_2\cup U_0},\dots,\cC_{\mb{f},U_t\cup U_0}$, the final $[tk+z,tk]_q$ code $\cC_F$ and the irregular conversion procedure $\varphi$. The minimum Hamming distance of $\cC_F$ is at least $tk+z-\deg(\Lambda_{\cup_{i\in[t]}\hat{U}_i}(x))+1$. Moreover, the read access cost is $tz$ and the write access cost is $z$. Furthermore, if $\sC$ is an MDS convertible code, then $\sC$ is per-symbol access-optimal.
\end{theorem}

\begin{IEEEproof}
The proof is a special case of Theorem \ref{thm9}. The minimum Hamming distance directly comes from $\deg(f(x))\le \deg(\Lambda_{\cup_{i\in[t]}\hat{U}_i}(x))-1$ in \eqref{eqn-degf}, which means there are at most $\deg(\Lambda_{\cup_{i\in[t]}\hat{U}_i}(x))-1$ roots of $f(x)$ over $\mathbb{F}_q[x]$. If $\sC$ is an MDS convertible code, then $\sC$ is per-symbol access-optimal, which directly comes from Theorem \ref{thm bound} and Definition \ref{def per}.
\end{IEEEproof}

The following corollary follows directly from Corollary \ref{cor-con C}.

\begin{corollary}
	Let $\mathbf{f} = (1, x, \dots, x^{k-1}) \in (\F_{q}[x])^k$ and $\hat{U}_i=U_i$ for $i\in[t]$ in Construction \ref{con A}. Then, the irregular convertible code $\sC$ converts $t$ Reed-Solomon codes into a final code $\mathcal{C}_F$ equivalent to a Reed-Solomon code $\mathcal{C}_{\hat{\mathbf{f}}=(1,x,\dots,x^{tk-1}),U_1\cup \cdots\cup U_t\cup U}$ with per-symbol access-optimal cost.
\end{corollary}

\subsection{Degraded Construction \ref{con D}}\label{sub-de-cond}

In this section, we show degraded Construction \ref{con D} and also give corresponding conclusions.

\begin{construction}[Degraded Construction \ref{con D}]\label{con B}
Let $t,\ell,\ell'$ are positive integers. Let $\F[x;\sigma,\delta]=\F_{q}[x;x^q,0]=\mathbb{F}_q[x]$ with $q\ge t\ell+\ell'+1$.
\begin{enumerate}
	\item \textbf{Choose P-basis.}
		\begin{enumerate}
			\item \textbf{Unchanged Symbols.} Let $U_i=\{ \alpha_{i,1},\alpha_{i,2},\dots,\alpha_{i,k} \}$  $i\in[t]$, where $k=\ell$.
			\item \textbf{Read Symbols.} Let $U_0=\{ \alpha_{0,1},\alpha_{0,2},\dots,\alpha_{0,k'} \}$, where $k'=\ell'$. Let $U_{0,i}=\{ \alpha_{0,i,1},\alpha_{0,i,2},\dots,\alpha_{0,i,r} \} \subseteq U_0$  for $i\in[t]\cup\{0\}$ and some $z\le\min\{\ell,\ell'\}$, where $\{\alpha_{0,i,j}\}_{j\in[z]} \subseteq \{\alpha_{0,j}\}_{j\in[\ell']}$ and $r=z$.
		\end{enumerate} 
	\item \textbf{Set Initial Codes.} Let $\mb{f}=(f_1,f_2,\dots,f_{k}) \in (\F_{q}[x])^{k}$ such that $f_1,f_2,\dots,f_{k}$ are linearly independent over $\mathbb{F}_{q}$. Set the initial code $\mathcal{C}_{\mathbf{f},U_1\cup U_0}$. 
	\item \textbf{Components of sCRT.} Let $U_{-i} = \bigcup_{j\in[t]\setminus\{i\}} U_j$ for $i\in[t]$. Choose any set $\hat{U}_i\subseteq \mathbb{F}_q$ for $i\in[t]$ satisfying
		\begin{equation*}
			U_i\cap \hat{U}_i= U_i,\ U_{-i}\cap \hat{U}_i=\emptyset
		\end{equation*}
		such that 
		\begin{equation*}
			\max_{i\in[k]}\{ \deg(f_i(x)) \} \le  \min_{i\in[t]}\{\deg(\Lambda_{\hat{U}_{i}}(x))\} - 1.
		\end{equation*}
	\item \textbf{Restricted Structures.} Suppose the following conditions hold:
		\begin{enumerate}
  			\item[\textbf{C1}:] $\mb{G}_{\mb{f},U_i}$ is an invertible matrix for $i\in[t]$ defined in Definition \ref{def pec};
  			\item[\textbf{C2}:] For $i\in[t]$, there exists an invertible diagonal matrix $\mb{D}_i=\diag(d_{i,1},d_{i,2},\dots,d_{i,k}) \in \F_{q}^{k\times k}$ such that
      				\begin{equation*}
        			\anglenv{ \mb{G}_{\mb{f},U_1} \mb{D}_i \mb{G}_{\mb{f},U_i}^{-1} \mb{G}_{\mb{f},U_{0,0}} } \subseteq \anglenv{\mb{G}_{\mb{f},U_{0,i}}}
      				\end{equation*}
      			i.e.,
      				\begin{equation*}
      				\mb{G}_{\mb{f},U_1} \mb{D}_i \mb{G}_{\mb{f},U_i}^{-1} \mb{G}_{\mb{f},U_{0,0}}
      				= \mb{G}_{\mb{f},U_{0,i}} \mb{L}_i,
      				\end{equation*}
     			where $\anglenv{\mb{G}_{\mb{f},U_{0,i}}}$ denotes the linear space spanned by the columns of $\mb{G}_{\mb{f},U_{0,i}}$ and  $\mb{L}_i = (\ell_{i,z,j})_{z,j\in[r]}$ is an $r\times r$ matrix for $i\in[t]$.
		\end{enumerate}
	\item \textbf{Conversion.} Let $\hat{U}_{-i}=\bigcup_{j\in[t]\setminus\{i\}}\hat{U}_j$ for $i\in[t]$. Define a conversion map $\varphi : \prod_{i=1}^{t} \cC_{\mb{f},U_1\cup U_0} \to \cC_F$ such that $\varphi(c_{\mb{m}_1},\dots,c_{\mb{m}_t}) = (\bar{c}_{\mb{m}_1},\dots,\bar{c}_{\mb{m}_{t}}, \bar{c})$, i.e.,
		\begin{equation*}
			\cC_F \triangleq  \left\{ \varphi(c_{\mb{m}_1},\dots,c_{\mb{m}_t}) :  c_{\mb{m}_i} \in \cC_{\mb{f},U_1\cup U_0}, i\in[t] \right\},
		\end{equation*}
		where $\bar{c}_{\mb{m}_i} = ((\mb{m}_i \mb{f}^\top)(\gamma))_{\gamma\in U_1}$ are unchanged symbols for $i\in[t]$, and $\bar{c} = (y_j(\gamma_{0,0,j}))_{j\in[r]}$ are written symbols with
			\begin{equation*}
 				y_j(\gamma_{0,0,j}) = \sum_{i=1}^{t} \parenv{ \sum_{z=1}^{r} \ell_{i,z,j} (\mb{m}_i \mb{f}^\top)(\gamma_{0,i,z}) }  \Lambda_{\hat{U}_{-i}}(\gamma_{0,0,j}).
			\end{equation*}
\end{enumerate}
\end{construction}

\begin{theorem}\label{thm4}
The map $\varphi$ in Construction \ref{con B} implies a $(t,1)_q$ irregular convertible code $\sC$ consisting of the initial code $\cC_{\mb{f},U_1 \cup U_0}$, the final $[tk+r,tk]_q$ code $\cC_F$ and the irregular conversion procedure $\varphi$. The minimum Hamming distance of $\cC_F$ is at least $r+tk-\deg(\Lambda_{\cup_{i\in[t]}\hat{U}_i}(x))+1$. Moreover, the read access cost is $tr$ and the write access cost is $r$.
\end{theorem}

\begin{IEEEproof}
The proof is a special case of Theorem \ref{thm10}. The minimum Hamming distance directly comes from $\deg(f(x))\le \deg(\Lambda_{\cup_{i\in[t]}\hat{U}_i}(x))-1$ in \eqref{eqn-deg f2}, which means there are at most $\deg(\Lambda_{\cup_{i\in[t]}\hat{U}_i}(x))-1$ roots of $f(x)$ over $\mathbb{F}_q[x]$.
\end{IEEEproof}

The following conclusion directly comes from Lemma \ref{thm11}.

\begin{corollary}\label{thm5}
Use the notation in Construction \ref{con B}. Let $B_1 = \{\alpha_1,\dots,\alpha_{k}\}$ and $B_2 = \{ \delta_1=1, \delta_2, \dots, \delta_{s}\}$ are subsets of $\F_q$. Assume that $\delta_1 B_1, \dots, \delta_{t} B_1, B_2$ are disjoint sets for some $t<s$. Set $\mb{D}_i=\mb{I}_{k}$, $\mathbf{L}_i=\mathbf{I}_r$, $U_i= \delta_{i} B_1$, $U_0 = B_2$, $U_{0,0} = \bigcap_{i=1}^{t} \delta_i B_2$ with $r=|U_{0,0}|\le\min\{k,s\}$ and $U_{0,i} = \delta_i^{-1} U_{0,0}$ for $i\in[t]$. If there exist invertible matrices $\mb{Q}_1,\mb{Q}_2,\dots,\mb{Q}_{t}$ such that
\begin{equation}\label{eqn20}
\begin{split}
\mb{Q}_i \mb{G}_{\mb{f},U_1} &= \mb{G}_{\mb{f}, \delta_i U_1}, \\
\mb{Q}_i \mb{G}_{\mb{f},U_{0,i}} &= \mb{G}_{\mb{f}, \delta_i U_{0,i}}, \\
\end{split}
\end{equation}
then $\mb{D}_i, \mb{L}_i, U_i, U_{0,i}$ for $i\in[t]$ and $U_{0,0}$ satisfy the condition \textbf{C2}. Moreover, if convertible code $\sC$ derived by Construction \ref{con B} is an MDS convertible code in this case, then it is per-symbol access-optimal.
\end{corollary}


It is worth noting that the sets $B_1$ and $B_2$ that satisfy the assumption in Corollary \ref{thm5} are easy to find.

\begin{corollary}\label{cor1}
In following two cases,
\begin{enumerate}
	\item (Subgroup-type) let $|\mathbb{F}|=q=p_1^{m_1}p_2^{m_2}\dots p_z^{m_z}+1$, where $p_1,p_2,\dots,p_z$ are distinct prime numbers. Assume that $z\ge 2$. There are two multiplication subgroups $D_1, D_2$ of $\F_q$ such that $D_1 = \{1,\alpha,\dots,\alpha^{p_1-1}\}$, $D_2 = \{1, \delta, \dots, \delta^{p_2-1}\}$ and $D_1 \cap D_2 = \{1\}$;
	\item (Basis-type) let $|\mathbb{F}|=q^{(k+1)(m+1)}$, where $q$ is a prime power. Suppose $D_1=\{1,\alpha,\cdots,\alpha^{k}\}$ is a basis of $\mathbb{F}_{q^{k+1}}$ over $\mathbb{F}_{q}$ and $D_2=\{1,\delta,\dots,\delta^{m}\}$ is a basis of $\mathbb{F}_{q^{(k+1)(m+1)}}$ over $\mathbb{F}_{q^{k+1}}$.
\end{enumerate}
Then, the sets $B_1 = D_1 \setminus \{1\}$ and $B_2 = D_2$ satisfy that $\delta_1 B_1, \dots, \delta_{t} B_1, B_2$ are disjoint sets, where $\delta_i = \delta^{i-1}$ for $i\in[t]$ with $t\le p_2$ for Subgroup-type and $t\le m+1$ for Basis-type. 
\end{corollary}

\begin{IEEEproof}
We prove the two cases separately.

\textbf{Case 1.} Since $D_1$ and $D_2$ are subgroups of $\mathbb{F}_q^*$ of distinct prime orders $p_1$ and $p_2$, respectively, we have $D_1\cap D_2 = \{1\}$. If $\delta_iB_1\cap\delta_jB_1\ne\emptyset$ for some $i\neq j$, then there exist $b,c\in D_1\setminus\{1\}$ such that $\delta_i b = \delta_j c$. Hence $\delta_i^{-1}\delta_j = b c^{-1}$. Therefore $\delta_i^{-1}\delta_j \in D_1\cap D_2 = \{1\}$ since $\delta_i^{-1}\delta_j\in D_2$ and $b c^{-1}\in D_1$. This implies a contradiction that $i=j$. Similarly, if $\delta_iB_1\cap B_2\ne\emptyset$, then $\delta_i b = d$ for some $b\in D_1\setminus\{1\}$ and $d\in D_2$, so $b = \delta_i^{-1}d \in D_1\cap D_2 = \{1\}$, contradicting $b\neq 1$. Thus $\delta_1B_1,\ldots,\delta_tB_1,B_2$ are pairwise disjoint. 

\textbf{Case 2.} Let $B_1 = D_1\setminus\{1\}$, $B_2 = D_2$. Suppose $\delta_iB_1\cap\delta_jB_1\ne\emptyset$ for some $i\neq j$. Then there exist $a,b\in\{1,\ldots,k\}$ such that $\delta^{i-1}\alpha^a = \delta^{j-1}\alpha^b$. Without loss of generality, assume $i>j$. Then, $1\le i-j\le t-1\le m$ and $\delta^{i-j}\in D_2$. Since $\delta^{i-j}=\alpha^{b-a}\in\mathbb{F}_{q^{k+1}}$ and $D_2$ is a basis of $\mathbb{F}_{q^{(k+1)(m+1)}}$ over $\mathbb{F}_{q^{k+1}}$, the element $\delta^{i-j}$ cannot lie in $\mathbb{F}_{q^{k+1}}$. Otherwise $1$ and $\delta^{i-j}$ would be linearly dependent over $\mathbb{F}_{q^{k+1}}$, contradicting the basis property. Hence the sets $\delta_iB_1,i\in[t]$ are pairwise disjoint.

Next, suppose $\delta_iB_1\cap B_2\ne\emptyset$ for some $i$. Then there exist $a\in\{1,\ldots,k\}$ and $r\in\{0,\ldots,m\}$ such that $\delta^{i-1}\alpha^a = \delta^r$. If $i-1\neq r$, then $\delta^{i-1}\alpha^a - \delta^r = 0$ is a nontrivial linear combination of the distinct basis elements $\delta^{i-1}$ and $\delta^r$ over $\mathbb{F}_{q^{k+1}}$, contradicting the basis property. If $i-1 = r$, then $\alpha^a=1$, which makes $\{1,\alpha,\ldots,\alpha^k\}$ linearly dependent, contradicting the fact that $D_1$ is a basis. Thus $\delta_iB_1\cap B_2=\emptyset$ for all $i$.

Consequently, the sets $\delta_1B_1,\ldots,\delta_tB_1,B_2$ are pairwise disjoint.
\end{IEEEproof}

\subsection{Comparison, Unification, and Gabidulin Application}

We first use the commutative specialization to relate the framework to known polynomial-form constructions. The purpose of this comparison is not to claim that the recovered Reed--Solomon and Tamo--Barg codes are new; rather, showing that they satisfy the compatibility conditions demonstrates that Construction~\ref{con B} captures previously separate polynomial-form designs within the same evaluation-level mechanism. The corresponding conditions from \cite{kong2024Locally} are recalled below.
\begin{enumerate}
\item[\textbf{C3:}] \cite[Construction I]{kong2024Locally} For $1\le j\le t$, $\mb{M}_i\cdot (1,\beta_{0,j},\dots,\beta_{0,j}^{k-1})^\top \in \mathrm{Span}_{\F_q}\{ (1,\beta_{0,j},\dots,\beta_{0,j}^{k-1})^\top : 1\le j\le t \}$.
\item[\textbf{C4:}] \cite[Construction III]{kong2024Locally} For every $2\le i\le t$, there is a $kr\times kr$ matrix $\mb{M}_i$ satisfying: $\mb{M}_i \cdot \mb{a}_{i,s,j} = \theta_{i,s} \mb{a}_{1,s,j}, (s,j)\in[k]\times[r+1]$ and $\mb{M}_i\cdot \mb{c}_{s,j}\in\mathrm{Span}_{\F_q}\{\mb{c}_{s',j'} : (s',j')\in[\ell]\times[r]\}, (s,j)\in[\ell]\times r$.
\end{enumerate}
Recall that $\mb{M}_i$ in \textbf{C3} and \textbf{C4} can be represented as $\mb{M}_i = \mb{N}\cdot\diag(\theta_{i,1},\theta_{i,2},\dots,\theta_{i,k})\cdot\mb{U}_i^{-1}$ by \cite[Lems.~II.3, III.3]{kong2024Locally}. Setting $\mb{D}_i=\diag(\theta_{i,1},\theta_{i,2},\dots,\theta_{i,k})$ and $\mb{G}_{\mathbf{f},U_i}=\mb{U}_i$ in Construction~\ref{con B}, condition \textbf{C1} follows from the invertibility of $\mb{U}_i$, while \textbf{C2} is implied by \textbf{C3} and \textbf{C4}. Hence the known Reed--Solomon and Tamo--Barg conversions are recovered as special cases of the commutative PEC template.

We next use the same template to obtain a code-family consequence not covered by the preceding Reed--Solomon/Tamo--Barg specializations. In the commutative ring $\mathbb{F}_{q^m}[x;x^{q^m},0]=\mathbb{F}_{q^m}[x]$, choosing a $q$-linearized polynomial basis leads to a Gabidulin instance. To abuse notation, for a set $B=\{b_1,b_2,\dots,b_m\} \subseteq \F_{q^m}$ such that $b_1,b_2,\dots,b_m$ are $\F_q$-linearly independent, denote $\anglenv{B}_{q} \triangleq \mathrm{Span}_{\F_q}\{b_1,b_2,\dots,b_m\}$. We first recall a property of a $\F_q$-linear subspace of $\F_{q^m}$.

\begin{lemma}[\cite{fang2018Deep}]\label{lem0}
Let $U$ be a $\F_q$-linear subspace of $\F_{q^m}$. If $\{b_1,b_2,\dots,b_n\}$ is a basis of $U$, then
\begin{equation*}
    \prod_{\alpha\in U}(x-\alpha) = \lambda
    \begin{vmatrix}
        b_1 & \cdots & b_n & x \\
        b_1^q & \cdots & b_n^q & x^q \\
        \vdots & & \vdots & \vdots \\
        b_1^{q^n} & \cdots & b_n^{q^n} & x^{q^n}
    \end{vmatrix},
\end{equation*}
where $\lambda\in\F_{q^m}$ is a constant determined by $b_1,b_2,\dots,b_n$.
\end{lemma}

The following corollary gives, to the best of our knowledge, the first merge-regime convertible construction in which the initial codes are Gabidulin codes and the final code is equivalent to a Gabidulin code. The novelty claim concerns the convertible-code construction and its symbol-access cost.

\begin{corollary}\label{thm6}
Suppose $D_1 = \{1,\alpha,\dots,\alpha^{k}\}$ is a basis of $\F_{q^{k+1}}$ over $\F_q$, and $D_2 = \{ 1, \gamma,\dots,\gamma^{m} \}$ is a basis of $\F_{q^{(k+1)(m+1)}}$ over $\F_{q^{k+1}}$, where $k+t-2\ge m>t$. Let $\mb{f}=(x,x^q,\dots,x^{q^{k-1}})\in (\F_{q^{(k+1)(m+1)}}[x])^k$. Then, there exists a polynomial convertible code given by Construction \ref{con B} such that the initial codes are Gabidulin codes and the final code is equivalent to a Gabidulin code. Moreover, this construction is per-symbol access-optimal.
\end{corollary}

\begin{IEEEproof}
Note that the product set $M=\{\alpha^a \gamma^b : 0\le a\le k, 0\le b\le m\}$ is a basis of $\mathbb{F}_{q^{(k+1)(m+1)}}$ over $\mathbb{F}_q$. Let $B_1=D_1\setminus\{1\} ,B_2=D_2$ and $\delta_i = \gamma^{i-1}$ for $i\in[t]$ be defined as in Corollary \ref{cor1} by Basis-type. Following Corollary \ref{thm5}, set $U_i=\delta_iB_1,U_0=B_2$,
	\begin{equation*}
		U_{0,0}= \bigcap_{i=1}^{t} \delta_i B_2= \bigcap_{i=1}^{t} \{\gamma^{i-1},\dots,\gamma^{i-1+m}\} = \{\gamma^{t-1}, \gamma^{t}, \cdots, \gamma^{m}\},
	\end{equation*}
 	$U_{0,i}=\delta_i^{-1}U_{0,0}$ and $\mb{D}_i=\mb{I}_{k}$, $\mathbf{L}_i=\mathbf{I}_r$ for $i\in[t]$. Then, $r=|U_{0,0}|=m-t+2\le\min\{k,m+1\}$ by $k+t-2\ge m$ and $U_{0,0}\subseteq B_2$.

Since $B_1$ is a subset of an $\mathbb{F}_q$-basis of $\mathbb{F}_{q^{k+1}}$, its elements are linearly independent over $\mathbb{F}_q$. Multiplication by the non-zero scalar $\delta_i$ preserves this property, so the evaluation points in $U_i = \delta_i B_1$ are $\mathbb{F}_q$-linearly independent. Hence $\mathbf{G}_{\mathbf{f}, U_i}$ is a full-rank Moore matrix. Therefore, the condition \textbf{C1} in Construction \ref{con B} holds. The condition \textbf{C2} comes from Corollary \ref{thm5} and Corollary \ref{cor1}, where $\delta_1B_1,\dots,\delta_tB_1,B_2$ are disjoint and $\mb{Q}_i = \diag(\delta_i,\delta_i^q,\dots,\delta_i^{q^{k-1}})$ satisfies \eqref{eqn20} for $i\in[t]$.

Note that $\delta_1B_1=U_1,\dots,\delta_tB_1=U_t,B_2=U_0$ are disjoint and $\delta_1B_1\cup\cdots\cup\delta_tB_1\cup  B_2$ is a subset of $M$. Thus, the initial codes are Gabidulin codes with the generator matrix $\mb{G}_{\mb{f},U_1\cup U_0}$ since $U_1\cup U_0$ are linearly independent over $\mathbb{F}_q$. Moreover, to prove that the final code in Construction \ref{con B} is equivalent to a Gabidulin code, it is sufficient to prove that $f(x)$ in \eqref{eqn37} has the form $\sum_{i=1}^{tk} a_i x^{q^{i-1}} \in \F_{q^{(k+1)(m+1)}}[x]$ since $U_{0,0}\subseteq U_0$ and $U_1\cup\cdots\cup U_t\cup U_{0,0}$ are linearly independent over $\mathbb{F}_q$.

For $i\in[t]$, suppose that
\begin{equation*}
    \Lambda_{\hat{U}_i}(x) = \prod_{\alpha \in \anglenv{U_i}_{q}\setminus\{0\}}( x - \alpha ).
\end{equation*}
Then, $\Lambda_{\hat{U}_1}(x),\dots,\Lambda_{\hat{U}_t}(x)$ are pairwise co-prime and
\begin{equation*}
    \deg(\Lambda_{\hat{U}_i}(x))= |\anglenv{U_i}_{q}\setminus\{0\}| = q^{k}-1.
\end{equation*}
Let $K_{-i} = \bigcup_{j\in[t]\setminus\{i\}} \anglenv{U_j}_q \setminus\{0\}$ for $i\in[t]$, and
\begin{equation}\label{eqn22}
    \Lambda_{\hat{U}_{-i}}(x) = \prod_{j\in[t]\setminus\{i\}} \Lambda_{\hat{U}_j}(x) = \prod_{\alpha \in K_{-i} } (x-\alpha)
\end{equation}
for $i\in[t]$. Denote $K_{-(i,j)}\triangleq \anglenv{U_1\cup\cdots\cup U_t\setminus\{\alpha_{i,j}\}}_q$ and $K_{i-(i,j)} \triangleq K_{-(i,j)} \setminus K_{-i}$ for $i\in[t],j\in[k]$. Then, $U_i\setminus\{\alpha_{i,j}\}\subseteq K_{i-(i,j)}$. We set
\begin{equation}\label{eqn23}
    e_i(x) = \sum_{j=1}^{k} \frac{ s_{i,j}  \prod_{ \alpha \in K_{i-(i,j)} } (x-\alpha) }{ \prod_{ \alpha \in K_{i-(i,j)} } (\alpha_{i,j}-\alpha) },
\end{equation}
where $s_{i,j} = (\mb{m}_i  \mb{f}^\top)(\alpha_{1,j})$ for $i\in[t], j\in[k]$, which implies that
\begin{equation*}
f(x) = \sum_{i=1}^{t} e_i(x) \Lambda_{\hat{U}_{-i}}(x)
\end{equation*}
and
\begin{equation*}
\begin{aligned}
    f(\alpha_{u,v})
    &=  e_u(\alpha_{u,v}) \Lambda_{\hat{U}_{-u}}(\alpha_{u,v}) \\
    &= \prod_{\alpha\in K_{-u}}(\alpha_{u,v}-\alpha) \cdot \sum_{j=1}^{k} \frac{ s_{u,j}  \prod_{ \alpha \in K_{u-(u,j)} } (\alpha_{u,v}-\alpha) }{ \prod_{ \alpha \in K_{u-(u,j)} } (\alpha_{u,j}-\alpha) } \\
    &=  s_{u,v} \prod_{\alpha\in K_{-u}}(\alpha_{u,v}-\alpha)
\end{aligned}
\end{equation*}
for $u\in[t],v\in[k]$, where the last equation holds by $\alpha_{u,v}\in U_u\setminus\{\alpha_{u,j}\}\subseteq K_{u-(u,j)}$ for all $j\in[k]\setminus\{v\}$ and $\alpha_{u,v}\not\in K_{u-(u,v)}$.

Denote $\det(\mb{M})$ as the determinant of a square matrix $\mb{M}$. Note that $K_{-(i,j)} = K_{i-(i,j)}\cup K_{-i}$. Combining \eqref{eqn22} and \eqref{eqn23}, we have that
\begin{align*}
e_i(x) \Lambda_{\hat{U}_{-i}}(x)
&= \sum_{j=1}^{k}  \frac{ s_{i,j} \prod_{ \alpha \in K_{-(i,j)}  } (x-\alpha) }{ \prod_{ \alpha \in K_{i-(i,j)} } (\alpha_{i,j}-\alpha) } \\
&= \sum_{j=1}^{k}   \frac{ s_{i,j} \lambda_{i,j} \det(\mb{G}_{\hat{\mb{f}},(U_1\cup\cdots\cup U_t\setminus\{\alpha_{i,j}\})\cup\{x\}} ) }{ \prod_{ \alpha \in K_{i-(i,j)} } (\alpha_{i,j}-\alpha) }
\end{align*}
where $\lambda_{i,j} \in\F_{q^{(k+1)(m+1)}}$ is some non-zero constant element determined by $U_1,U_2,\dots,U_t$ for $i\in[t],j\in[k]$, and $\hat{\mb{f}}=(x,x^q,\dots,x^{q^{tk-1}})\in (\F_{q^{(k+1)(m+1)}}[x])^{tk}$, and the second equation holds by Lemma \ref{lem0}. Note that
\begin{equation*}
\det(\mb{G}_{\hat{\mb{f}},(U_1\cup\cdots\cup U_t\setminus\{\alpha_{i,j}\})\cup\{x\}}) = \sum_{\ell=1}^{tk} b_\ell x^{q^{\ell-1}}\in \F_{q^{(k+1)(m+1)}}[x],
\end{equation*}
where $b_i$ is determined by $U_1,\dots,U_t$ for $i\in[tk]$. This is to say that
\begin{equation*}
f(x) = \sum_{i=1}^{t} e_i(x)\Lambda_{\hat{U}_{-i}}(x) = \mb{m} \hat{\mb{f}}^\top,
\end{equation*}
where $\mb{m}\in\F_{q^{(k+1)(m+1)}}^{tk}$. Therefore, the final code is equivalent to the Gabidulin code $\tilde{\cC}$. Since Gabidulin codes are MDS codes,  the per-symbol access-optimal property directly comes from Corollary \ref{thm5}.
\end{IEEEproof}

\begin{remark}
On the one hand, the per-symbol access-optimality in Corollaries~\ref{cor-con C},\ref{cor5},\ref{thm6} is with respect to the symbol-access model of Definition~\ref{def per}, obtained through the MDS property. The present result does not by itself establish that the conversion map preserves rank-metric or sum-rank-metric structure; that question is left for future work. On the other hand, the code equivalence mentioned in this paper preserves the rank metric for Corollaries~\ref{cor-con C},\ref{cor5}, since the coordinate scaling factors $\Lambda_{\hat{U}_{-i}}(\gamma)\in\mathbb{F}_q$ in \eqref{eqn-scalar} and $d_{i,j}\Lambda_{\hat{U}_{-i}}(\gamma_{i,j})\in\mathbb{F}_q$ \eqref{eqn35} by Lemma \ref{lem4}-2) and $\mb{D}_i=\mb{I}_k,i\in[t]$. However, in the degenerate case over $\mathbb{F}_{q^{(k+1)(m+1)}}[x]$, the coordinate scaling factors for Corollary \ref{thm6} involve evaluations of the corresponding annihilating polynomials and generally lie in the base field $\mathbb{F}_{q^{(k+1)(m+1)}}$, rather than necessarily in $\mathbb{F}_q$. Hence, the equivalence preserves Hamming weight and the MDS property, but it does not necessarily preserve rank weight. Constructing a merge-regime conversion with Gabidulin initial codes and a final code that is rank-metric equivalent to a Gabidulin code remains an interesting open problem.
\end{remark}

\section{Conclusion}\label{sec con}

This paper develops a polynomial-evaluation view of merge-regime code conversion. The main contribution is a common algebraic mechanism that combines local message-polynomial behavior into a final evaluation polynomial while keeping explicit control of unchanged and written symbols. In the skew ring $\mathbb{F}_{q^m}[x;x^q,0]$, the required mechanism is provided by minimal skew polynomials of conjugacy-class unions together with an evaluation-compatible product rule and the sCRT of Theorem~\ref{thm8}. This perspective leads to two skew-PEC conversion templates: Construction~\ref{con C} for distinct initial codes and Construction~\ref{con D} for identical initial codes under the explicit compatibility conditions \textbf{C1} and \textbf{C2}.

For the standard skew-polynomial basis, the constructions preserve the linearized Reed--Solomon evaluation structure up to code equivalence and attain per-symbol access-optimal cost. Specializing the framework to $\mathbb{F}_q[x]$ recovers the known polynomial-form Reed--Solomon and Tamo--Barg constructions as instances of the same template. The same commutative template also yields, to the best of our knowledge, the first merge-regime construction with Gabidulin initial codes and a final code equivalent to a Gabidulin code, again with optimal symbol access under Definition~\ref{def per}.

The framework also makes its current limitations explicit. The skew constructions require complete P-bases of conjugacy classes, resulting in the divisibility condition $m\mid k$; their field-size requirements are not intended to improve the best known generic MDS-conversion bounds; and the optimality statements for the LRS and Gabidulin instances concern the symbol-access model used in this paper rather than a new rank/sum-rank conversion metric. Important directions include relaxing the P-basis/divisibility restriction, characterizing broader classes of PECs that satisfy \textbf{C1}--\textbf{C2}, determining whether stronger rank- or sum-rank-metric preservation holds for the structured instances, and extending the evaluation-compatible approach to more general skew rings $\F[x;\sigma,\delta]$.

\bibliographystyle{IEEEtran}
\bibliography{GeBib}
\end{document}